\documentclass[conference]{IEEEtran}
\title{EDCCC2017amREady}
\title{Preserving Stabilization while {\em Practically} Bounding State Space}

\author{\IEEEauthorblockN{Vidhya Tekken Valapil\IEEEauthorrefmark{1},
Sandeep S. Kulkarni\IEEEauthorrefmark{2}}
\IEEEauthorblockA{Department of Computer Science,
Michigan State University\\
East Lansing, Michigan, USA. 48824\\
Email: \IEEEauthorrefmark{1}tekkenva@cse.msu.edu,
\IEEEauthorrefmark{2}sandeep@cse.msu.edu}}

\usepackage{graphicx}
\usepackage{amsmath}

\usepackage{amsthm}
\usepackage{algorithm}
\usepackage{subfigure}
\usepackage[noend]{algpseudocode}
\usepackage{xspace}
\usepackage{comment}
\usepackage{cite}
\theoremstyle{definition}
\newtheorem{definition}{Definition}[section]

\theoremstyle{definition}
\newtheorem{remark}{Remark}[section]

\usepackage[utf8]{inputenc}
\begin{document}

\maketitle
\newcommand{\regionsize}{\ensuremath{\mathcal{RS}}\xspace}
\newcommand{\SK}[1]{\footnote{Note from Sandeep #1}}
\newcommand{\minregion}{\ensuremath{r_{min}}\xspace}
\newcommand{\maxregion}{\ensuremath{r_{max}}\xspace}
\newcommand{\maxinc}{\ensuremath{max_{inc}}\xspace}

\newcommand{\br}[1]{\ensuremath{\langle #1 \rangle}\xspace}
\newcommand{\maxdep}{\ensuremath{v}\xspace}

\newcommand{\freelegbound}[1]{\ensuremath{[3*#1*\maxinc..3*(#1+1)*\maxinc+2*\maxinc-1]}\xspace}

\newcommand{\minfreecounter}[1]{\ensuremath{3*#1*\maxinc}\xspace}
\newcommand{\maxfreecounter}[1]{\ensuremath{3*(#1+1)*\maxinc + 2*\maxinc -1}\xspace}

\newcommand{\maxv}{max_r}
\newcommand{\maxr}{max_r}

\newcommand{\mindepcounter}[1]{\minfreecounter{(#1-2-\maxv)}\xspace}
\newcommand{\maxdepcounter}[1]{\maxfreecounter{#1}\xspace}

\newcommand{\maxrange}{\ensuremath {\maxinc*(11 + 3*\maxv)}}

\newcommand{\maxbound}{3*[\maxrange]\xspace}
\newcommand{\intervallength}{\ensuremath{\maxinc*(11+3*\maxv)}\xspace}

\newcommand{\newfloor}[1]{\ensuremath{\lfloor #1 \rfloor}\xspace}

\newcommand{\beforestepdistance}{\ensuremath{k_b}\xspace}
\newcommand{\afterstepdistance}{\ensuremath{k_f}\xspace}
\newcommand{\beforeregiondistance}{\ensuremath{r_b}\xspace}
\newcommand{\afterregiondistance}{\ensuremath{r_f}\xspace}

\newtheorem{theorem}{Theorem}
\newtheorem{lemma}{Lemma}

\newcommand{\regionzero}{region $0$\xspace}
\newcommand{\regionone}{region $1$\xspace}
\newcommand{\regiontwo}{region $2$\xspace}

\newcommand{\regionrr}{region $r$\xspace}
\newcommand{\regionrpone}{region $r+1$\xspace}
\newcommand{\regionrptwo}{region $r+2$\xspace}

\newcommand{\maxsteps}{\ensuremath{v}\xspace}

    % See p.105 of "TeX Unbound" for suggested values.
    % See pp. 199-200 of Lamport's "LaTeX" book for details.
    %   General parameters, for ALL pages:
    \renewcommand{\topfraction}{0.9}	% max fraction of floats at top
    \renewcommand{\bottomfraction}{0.8}	% max fraction of floats at bottom
    %   Parameters for TEXT pages (not float pages):
    \setcounter{topnumber}{2}
    \setcounter{bottomnumber}{2}
    \setcounter{totalnumber}{4}     % 2 may work better
    \setcounter{dbltopnumber}{2}    % for 2-column pages
    \renewcommand{\dbltopfraction}{0.95}	% fit big float above 2-col. text
    \renewcommand{\textfraction}{0.07}	% allow minimal text w. figs
    %   Parameters for FLOAT pages (not text pages):
    \renewcommand{\floatpagefraction}{0.7}	% require fuller float pages
	% N.B.: floatpagefraction MUST be less than topfraction !!
    \renewcommand{\dblfloatpagefraction}{0.7}	% require fuller float pages

\begin{abstract}

%In this paper, we present an algorithm that utilizes partially synchronized physical clocks and properties of the given stabilizing algorithm that uses variables with unbounded domain into a stabilizing program that only uses variables with bounded domain. 
Stabilization is a key dependability property for dealing with unanticipated transient faults, as it guarantees that even in the presence of such faults, the system will recover to states where it satisfies its specification. One of the desirable attributes of stabilization is the use of bounded space for each variable. 

In this paper, we present an algorithm that transforms a stabilizing program that uses variables with unbounded domain into a stabilizing program that uses bounded variables and (practically bounded) physical time. 
  While non-stabilizing programs (that do not handle transient faults) can deal with unbounded variables by assigning {\em large enough but bounded} space, stabilizing programs --that need to deal with arbitrary transient faults-- cannot do the same since a transient fault may corrupt the variable to its maximum value. 
%
%Our transformation is based on two key concepts: free counters and dependent counters. The former represents variables that can be freely increased without affecting the correctness of the underlying program and the latter represents temporary variables that become irrelevant after some duration of time. 

We show that our transformation algorithm is applicable to several problems including logical clocks, vector clocks, mutual exclusion, leader election, diffusing computations, Paxos based consensus, and so on. Moreover, our approach can also be used to bound counters used in an earlier work by Katz and Perry for adding stabilization to a non-stabilizing program.   
By combining our algorithm with that earlier work by Katz and Perry, it would be possible to provide stabilization for a rich class of problems, by assigning {\em large enough but bounded} space for variables. 

\end{abstract}
\label{sec:intro}
\section{Introduction}

Self stabilization is one of the highly desirable dependability properties of distributed systems. Self-stabilization ensures that a system affected by a fault eventually stabilizes to or reaches a valid state in finite time. Stabilizing fault-tolerance is especially useful for dealing with unexpected transient faults. Transient faults can perturb the system to potentially arbitrary state, and guaranteeing that the program recovers to legitimate states ensures that the effect of these faults would only be temporary. 

A key desirable property of stabilizing systems is for them to utilize only variables with a bounded domain. For non-stabilizing systems (i.e. systems that do not handle transient faults), one could utilize counters that grow unbounded by ensuring that the value of the variable remains manageable during the length of the system computation. For example, one could argue that if a variable increases by at most $10$ every second and the system can run for at most $1000$ seconds then the value would never be more than $10,000$. However, for stabilizing systems with variables with a bounded domain, this argument does not hold true. This is because a transient fault could perturb the system to a state where the value has already reached $10,000$ i.e. the bound. The same argument was made by Lamport and Lynch \cite{LamportL90HandbookOfTheoCS} and we recall the argument in the following quote,

%\newpage

\begin{quote}
"Simply bounding the number of instance identifiers is of little practical significance, since practical bounds on an unbounded number of identifiers are easy to find. For example, with 64-bit identifiers, a system that chooses ten per second and was started at the beginning of the universe would not run out of identifiers for several billion more years. However, through a transient error, a node might choose too large an identifier, causing the system to run out of identifiers billions of years too soon--perhaps within a few seconds. A self-stabilizing algorithm using a finite number of identifiers would be quite useful, but we know of no such algorithm."
\end{quote}

However, the above quote about unbounded counters conflicts with usage in distributed systems where one often utilizes {\em time} as a variable whose value is theoretically unbounded. This is because with time there is a guarantee for convergence offered by protocols like NTP, i.e. any inconsistency can be easily detected and corrected in finite time. 

Based on this conflict --we seem to find the use of physical time (that is theoretically unbounded but practically bounded) acceptable but find the use of other unbounded variables unacceptable, so we consider the question: {\em Why is the usage of unbounded time reasonable, but the usage of other unbounded variables is not?}
  We observe that there is an inherent difference between the variable {\em time} and any other variable. In particular, detecting whether time is corrupted is much easier than detecting if other variables are corrupted. With the usage of redundancy, atomic clocks, etc., one can ensure that the time of a process is {\em close} to the correct value. In other words, if transient faults perturb a clock to a value that is far away from the current value, this corruption can be detected before using that clock value. 
  \begin{comment}
Also, if a process receives a message where the value of time is far away from its own clock (e.g., a message with time April 15, 3017/1017 instead of April 15, 2017), it can be detected before the receiving process accepts that message. 
\end{comment}

Observe that this property of time may not be satisfied by other variables in a given program. For example, if we use logical clocks by Lamport \cite{lamport}, it is possible that clocks of two processes could genuinely differ by a large value.

Our goal in this paper is to identify a class of programs for which we can begin with a stabilizing program that relies on unbounded counters and transform it into a program with bounded counters and (theoretically unbounded but practically bounded) physical time. 

\textbf{Contributions of the paper. }
\begin{itemize}
\item We introduce the notion of free and dependent counters and utilize them to develop an algorithm that transforms a stabilizing program with unbounded counters into a stabilizing program with bounded counters. 
\item We demonstrate that our approach can be combined with that by Katz and Perry \cite{Katz1993}. Specifically, \cite{Katz1993} provides a mechanism to transform a \textit{non-stabilizing} program into a stabilizing program with unbounded counters. We show that the generated stabilizing program can then be transformed into a stabilizing program that uses bounded counters and (practically bounded) physical clock. 
\item We demonstrate our algorithm in the context of several classical problems such as consensus, logical/vector clocks, mutual exclusion, diffusing computation, etc. 
\item We show that even with trivially satisfiable parameters for a practical system (like clock drift of less than $100$ seconds, messages are either delivered or lost within an hour, etc), the size of counters in our programs is small. 
\end{itemize}

\begin{comment}

We present an algorithm for transforming a stabilizing program with unbounded counters into a program with bounded counters. The main idea behind this algorithm is to introduce the notion of free counters and dependent counters. Intuitively, free counters are those that are independent of other counters and can increase at will. Dependent counters on the other hand cannot be changed at will. However, they have a limited life after which they no longer affect the operation of the program.

We observe that in fact several programs can be viewed to consist of a set of bounded variables, unbounded free counters and unbounded dependent counters. Examples of such programs are logical clocks, diffusing computations, paxos based consensus algorithms, leader election algorithms, mutual exclusion algorithms and so on. We show that for these and for several other programs, we can transform a given solution that uses unbounded counters into a solution that uses only bounded counters and physical time that is kept reasonably synchronized with NTP. 
\end{comment}

{\bf Organization of the paper. }
The rest of the paper is organized as follows: We define distributed programs and relate their execution with time in Section \ref{sec:preliminaries}. 
We define the notion of free and dependent counters in Section \ref{sec:freeanddependent} and illustrate them with the example of Lamport's logical clocks in Section \ref{sec:lamport}.   We present our algorithm in Section \ref{sec:algorithm} and discuss some applications of it in Section \ref{sec:application} (and some in Appendix \ref{sec:diff}, \ref{sec:vc} and \ref{sec:mutualexclusion}). Section \ref{sec:application} also demonstrates that our approach can be combined with the approach in \cite{Katz1993} by Katz and Perry, so that an existing program can be transformed into a stabilizing program that uses bounded variables and physical time.  Section \ref{sec:related} discusses related work and questions raised by our work. In Section \ref{sec:concl}, we present concluding remarks and future work. 
%Due to reasons of space, 
%In Appendix \ref{subsec:proof_lamport} we present the proofs related to the  illustration of our algorithm in Logical clocks from Section \ref{sec:lamport}. 
In Appendix \ref{sec:alg_illust} and \ref{sec:ProofofCorrectness}, we present the step-by-step illustration of our algorithm and its proof of correctness. 
In Appendix \ref{sec:SummaryofNotations} we summarize the notations used in this paper. 

\label{sec:Distributed Programs with Counters:}
\section{Preliminaries}
\label{sec:preliminaries}
\subsection{Modeling Distributed Programs}
\label{sec:modeldistributed}

%We target distributed programs that use unbounded counters, specifically two types of counters: free counters and dependent counters.

%Consider a system of $n$ processes $p_0,p_1,p_2, \cdots,p_n$ where each process $p_j$ has a set of bounded and unbounded variables (counters) and a set of guarded commands that update those variables and/or add new unbounded variables.

\noindent A distributed program consists of a set of processes. 
Each process has a set of actions and the program executes in an interleaving manner where an action of some process is executed in every step. 
%
%\footnote{Put this comment in the proof where needed. We ignore fairness requirement associated with program actions. This is due to the fact that we focus on transformation of a given program $p$ into a program $p'$. And, we show that behaviors of $p'$ can be mapped to behaviors of $p$. Hence, if $p$ requires fairness for correctness, $p'$ would be correct under the same fairness. }
%
%This execution is captured with a set of initial variables, say $V_j^{0}$, for process $j$. The union of these variables, $V_p^{0} = \cup V_j^{0}$ denotes the initial variables of program $p$. As an action of a process is executed, new variables may be introduced. Hence, we utilize $V_p^n$ to denote the variables of program $p$ in Step $n$. 
%
Execution of the program is captured with a set of program variables, each of which is associated with a domain.
%We assume that the domain is either finite or the set of natural numbers. 
\iffalse
Variables of the program can be classified as simple and compound variables. 
Simple variables include bounded and unbounded variables (counters). Complex (composite or abstract) variables include containers, lists, queues, etc.
The total number of variables in the program remains fixed throughout the execution.
However, when an action of a process is executed, new entries may be added to or removed from one or more complex variables. 
\fi
With this intuition, we now formally define a program in terms of its variables and actions. Definitions \ref{def:program} to \ref{def:stabilization} are from standard literature such as \cite{ghosh2014distributed,distributedSnapshots,dij}.

\theoremstyle{definition}
\begin{definition}{\bf{(Program).}} 
\label{def:program}
A program $p$ is of the form $\langle V_p,A_p \rangle$, where $V_p$ is a set of %simple (bounded variables and unbounded variables/counters) and complex
variables, and $A_p$ is a set of actions 
that are of the form $guard \rightarrow statement$, where $guard$ is a condition involving the variables in $V_p$ and the $statement$ updates a subset of variables in $V_p$.
\end{definition}

\begin{definition}{\bf{(State).}} 
A state $s$ of program $p$ is obtained by assigning each variable in $V_p$ a value from its domain. 
\end{definition}

\begin{definition}{\bf{(Enabled).}}
An action of the form $guard \rightarrow statement$ is enabled in state $s$ iff $guard$ evaluates to true in state $s$.
\end{definition}

\begin{definition}{\bf{(Computation).}}
A computation is a sequence of states $s_0, s_1, s_2,\cdots$, where a state $s_{l+1}$, $l \geq 0$, is obtained by executing some enabled action in state $s_l$. %an action(or actions) at state $s_l$. 
\end{definition}

\begin{remark}
For the sake of simplicity, we assume that there is at least one action enabled in state $s$. If such an action does not exist, we pretend that the program has an action corresponding to a self-loop at state $s$. 
\end{remark}

\begin{definition}{\bf{(Computation Prefix).}}
A finite sequence of states $s_0, s_1, s_2,$ $\cdots, s_n$ is a computation-prefix of program $p$ iff it is a prefix of some computation of $p$. 
\end{definition}
\iffalse
\begin{definition}{\bf{Computation Prefix.}}
A finite computation $M_f = s_0, s_1, s_2, \cdots, s_n$ is a computation prefix of a computation $M = s_0, s_1, s_2,\cdots$ iff:\footnote{Chec the above definition}

(i) The initial state of $M_f$ is same as the initial state of $M$,

(ii) If $s_n$ is the final state of $M_f$, then every state in $M_f$ from the initial state to $s_n$ is also present in $M$ in the same relative order or sequence as in $M_f$.
\end{definition}
In the program considered in this paper we partition variables in $V_p$ into simple variables and complex variables.
Simple variables are bounded variables with a finite domain or unbounded variables whose domain is $N$, the set of natural numbers.
A complex variable is an dynamic ''collection'' of simple variables, where a collection can be a set, queue, array, etc of flexible size.
We can view such a program by ignoring the structure provided by complex variables. 
In other words, if there is a complex variable $C \in V_p$ which is a collection or a set of simple variables say $C=\{c_1,c_2,c_3,\cdots\}$, then we can ignore (the structure provided by) $C$ by just adding its elements i.e. the simple variables $c_1,c_2,c_3,\cdots$ in $C$ to $V_p$ directly.
However, such an inclusion makes $V_p$ a dynamic set, whose size is not fixed anymore.
Given the mapping of complex variables to simple variables, we can view $V_p$ as a collection of only simple variables. So for simplicity, in our further discussion we will view the variables of the program in this manner.
\fi
\noindent Finally, we recall the definition of stabilization from \cite{dij}:
\begin{definition}
\label{def:stabilization}{\bf{(Stabilization).}} 
Program $p$ is stabilizing to $S$, where $S$ is a set of states, iff
\begin{itemize}
\noindent\item Starting from an arbitrary state, every computation of $p$ reaches $S$, and
\item Starting from a state in $S$, no computation of $p$ reaches a state outside $S$.
\end{itemize}
\end{definition}
%\subsection{Clock Synchronization and }
\subsection{Relating  Program Computation  and Time}
\label{sec:sysmodel}
%\begin{comment}
As discussed in the Introduction, our goal is to combine the existence of (reasonably synchronized) global time achieved through services such as NTP with reasonable timing properties in the given algorithm. Since the definitions in Section \ref{sec:modeldistributed} are time-independent, in this section, we identify the role of time and the relation between program steps and time in our algorithm. 
%In this section, we identify how these assumptions are represented in our algorithm. In particular, this model identifies the relation between {\em program steps} and {\em time}.
%\end{comment}
%Since we intend to use practically bounded physical clock to bound the variables in a stabilizing program in this section, we need to relate the program steps with time.
%In our algorithm, the relation between program steps and time is based on (1) reasonably synchronized clocks and (2) bound on the maximum number of steps that could be executed in a given time. 

Our algorithm relies on NTP-like algorithm to provide physical clock for each process, which is \textit{close to an abstract global clock} (this global clock is not available to processes themselves). We partition the abstract global time say $t$ into regions of size $\regionsize$. Thus, the (global) region is identified by $\newfloor{\frac{t}{\regionsize}}$. Likewise, each process $j$ is also associated with a physical time, say $t_j$. This time is also mapped to the region of process $j$. Thus, the region of process $j$ is $\newfloor{\frac{t_j}{\regionsize}}$. 
Note that due to clock drift the global region and region associated with process $j$ may not be identical. Likewise, region associated with process $j$ may not be the same as that associated with process $k$. 

We \textit{choose \regionsize such that} (1) the region identified by the process from its own local physical clock differs from the region identified by the global clock by at most $1$, and (2) the regions identified by two processes from their local clocks differ by at most $1$. 
%(1) the global region differs from the region associated with process $j$ by at most $1$, and (2) the regions associated with processes $j$ and $k$ differ by at most $1$. Clearly, if $\regionsize$ is chosen to be $1$ hour, this would be trivial. 
Given the current technology, choosing $\regionsize$ to be a few milliseconds would be reasonable for many existing systems to satisfy this assumption.  In our analysis, we assume \regionsize to be $100$ seconds. Note that achieving clock synchronization to be within $100$ seconds is trivial in any practical system. 

%Regarding the second relation, we recall that program computations is a sequence of states. 
%We also associate an (abstract) global time with each state. Thus, 

For a given computation, we identify a program subsequence that occurred in a given (abstract) global region. 
%we partition the program computations into regions based on the (abstract) global time identifying when that computation occurred.
Although the processes themselves are not aware of this (abstract) global time, this association allows us to model assumptions such as any message would be delivered within time $\delta$ or it will be lost.
\begin{comment}
2) any request for mutual exclusion would be satisfied within time $\delta$, etc.
\end{comment} 
We can model such assumptions in terms of regions; 
if we utilize regions to be $100$ second long and we are guaranteed that messages would either be received or lost within one hour (3600 seconds) then this would mean that a message has a lifetime of at most $36$ regions. 

\section{Free Counters and Dependent Counters}
\label{sec:freeanddependent}
In this section, we define the notion of free and dependent counters that form the basis of our transformation algorithm. 
%Next, we want to characterize some variables of the given program into {\bf free counters} and {\bf dependent counters}. 
However, before we do that, we focus on the structure of the variables in the program. In particular, for program $p$, we partition its variables $V_p$ into two types: simple variables and complex variables. Simple variables are those variables with domain that is either a finite set or $N$, the set of natural numbers. And, complex variables are {\em collections} (e.g., set, sequence, list, etc.) of simple variables, and the constituent variables can be removed/added dynamically. To define the notion of free and dependent counters, we will unravel the structure of a complex variable and focus only on the simple variables contained in it. For example, if the program contains a complex variable, say $C$, which is a set and its current value is $\{ 3, 5, 7\}$, then we visualize this as having three simple variables $c_1$, $c_2$ and $c_3$ whose values are 3, 5 and 7 respectively. 

With this intuition, we can view a program $p$ with variables $V_p$ as an equivalent program with variables $SV_p$,  where $SV_p$  is a dynamically changing collection of simple variables.
%However, as the program executes, the set $SV_p$ changes. For example, if a value is added to or removed from the set $C$ above then it would be equivalent to adding/removing a new (simple) variable in $SV_p$.
Moreover, the domain of any variable in $SV_p$ is either finite or equal to $N$, the set  of natural numbers. 
%The notion of free and dependent counters focuses on the variables in $SV_p$ with domain $N$. 

\begin{remark}
A reader might wonder why we do not define program $p$ in terms of $SV_p$ in the first place.
As mentioned above, $SV_p$ is a dynamic set that has a flexible size. 
To update a dynamic set of variables one would require a dynamic or infinite set of actions. Without making explicit efforts, such a model has the potential to model programs that are not recursively enumerable. 
Our modeling with complex variables in $V_p$ avoids this problem, as the set of actions is always finite. 
%So to avoid the need for dynamic set of actions, we define program $p$ to have a fixed-sized set of simple and complex variables $V_p$, where the complex variables are collections of flexible size. We illustrate this issue in Section \ref{sec:lamport} with an example. 

\end{remark}

The set $SV_p$ is dynamic.
%i.e. new simple variables may be added or removed from $SV_p$. 
We say that a variable in $SV_p$ is a {\bf{permanent variable}} if it is guaranteed to be present in every state of $p$. For example, any simple variable in $V_p$ would be a permanent variable since it will be present in $SV_p$ at all times. A variable that is not a permanent variable is called a {\bf{temporary variable}}. 

\theoremstyle{definition}\label{def:Valuation of variable in $V_p$}
\begin{definition}{\bf{(Valuation of variable in $V_p$).}} Let $x$ be a variable in $V_p$ and let $s$ be a state of program $p$. $x(s)$ denotes the value of $x$ in state $s$. 
\end{definition}

We overload this definition for $SV_p$. Specifically, if variable $x$ is present in $SV_p$ in the given state, the value of that variable is defined in the same manner as in the above definition. And, if the variable $x$ is not present in that state (entries in complex variables in $V_p$ or their equivalent simple variables in $SV_p$ may be added/removed), we denote its value as $\bot$. In other words,
 
\theoremstyle{definition}\label{def:Valuation of variable in $SV_p$}
\begin{definition}{\bf{(Valuation of variable in $SV_p$).}} Let $x$ be a variable in $SV_p$ and let $s$ be a state of program $p$. If $x$ is present in state $s$, then $x(s)$ denotes the value of $x$ in state $s$. And, if $x$ is not present in $s$, then we denote it as $x(s) = \bot$.
\end{definition}
\noindent With the help of $SV_p$ and permanent/temporary variables, we define the notion of free and dependent counters. Intuitively, a free counter is a permanent variable whose value never decreases. Moreover, if we increase the value of the free counter in the final state of a computation-prefix then the resulting sequence is also a valid computation prefix of the given program. Formally,

\begin{definition}\label{def:free_counter}{\bf{(Free counter).}} A permanent variable $fc$ of program $p$ is a free counter iff for any computation prefix $\rho = s_0, s_1, s_2,\cdots,s_l$ of $p$ the following conditions hold:

(i) $\forall w: 0 \leq w < l : fc(s_{w+1}) \geq fc(s_w)$,

(ii) %if $\rho = s_0, s_1, s_2,\cdots,s_l$ is a valid computation of $p$ then 
$\rho' = \rho + s_{l+1}$ is also a valid computation prefix, where state $s_{l+1}$ is reached from state $s_l$ by increasing the value of $fc$ (and leaving other variables unchanged), and $\rho + s_{l+1}$ denotes concatenation of $\rho$ and $s_{l+1}$.

%is the next state of $s_l$ where $fc(s_{l+1}) = fc(s_l) +d$.
% * <sandeepkulkarnimsu@gmail.com> 2016-11-18T16:43:10.589Z:
%
% Please see revised definition. The text that follows does not read well. Can you revise it. 
%
% ^.
\end{definition}
%{\bf Remove this paragraph????}
Thus, if $fc(s_l)$ is the value of the free counter $fc$ in program $p$ in state $s_l$, then $fc(s_{l+1})$ (i.e., the value of the free counter $fc$ in the subsequent state $s_{l+1}$) is never less than its value in the previous state $s_l$. 
Also, if $\rho = s_0, s_1, s_2,\cdots,s_l$ is a valid computation prefix of $p$, appending state $s_{l+1}$ where $fc(s_{l+1}) = fc(s_l) +d$ (where $d\geq0$) to $\rho$ results in another valid computation prefix $\rho' = \rho + s_{l+1}$.
%
%In other words, given that $fc(s_l)$ (free counter of program $p$ at state $s_l$) satisfies the system correctness, then in the subsequent state $s_{l+1}$ an increase in the value of $fc$ by any (positive) constant value $d$ continues to preserve the overall correctness of the system.

%\newcommand{\kdist}{\ensuremath{v_1}}
%\newcommand{\vdist}{\ensuremath{v_2}}
Next, we define the notion of dependent counters. A dependent counter is a temporary variable. We require that when this variable is created/added, its value is set to the value of some free counter within at most  $\beforestepdistance$ preceding steps.
Moreover, after $\afterstepdistance$ steps,  this temporary variable is removed.
And, in between the value remains unchanged. 

\begin{remark}
%\footnote{
Note that this requirement is not restrictive, because essentially, the requirement is just that the value assigned to the dependent counter is somehow related to a free counter in the recent past. For example, if variable $dc$ is set to $fc-5$ where $fc$ is a free counter, then we can treat it as having two variables $dc1$ and $dc2$, where setting $dc$ to $fc-5$ is modeled as setting $dc1$ to be same as $fc$ and $dc2$ to $-5$, and using $dc1+dc2$ instead of $dc$. Note that the latter is a bounded variable whereas the former can be used to satisfy the requirements of dependent counters. Likewise, setting $dc$ to $2*fc$ or $fc^2 + 10$ would be acceptable as well. Since there are too many such choices, to keep the transformation algorithm simple, we use the above definition. However, in practice it may require some {\em syntactic} tweaking of a given program without affecting its properties. 
%}
\end{remark}

\begin{remark}
%A reader may wonder about the restriction %that may be caused due to the fact 
%that a dependent counter cannot change its value and must be reset to $\bot$ within certain number of steps. 
The goal of this requirement is that the value of the counter will eventually become obsolete and hence will no longer affect the program execution. 
%For example, if $dc$ is a variable where we change $dc$ from $3$ to $5$, we may still be able to treat it as a dependent counter by splitting the update of $dc$ in two steps; first to change it form $3$ to $\bot$ whereby the value of the old variable is obsolete and then by adding a new variable (with the same name) which is initialized to $5$. 
We discuss this further in Section \ref{sec:application}, where this requirement is handled by syntactic changes to a given program.
%we show that this condition is not restrictive and can be satisfied by setting/reseting variables of the given program in a slightly different manner.
\end{remark}

\begin{definition}\label{def:dep_counter}{\bf{((Step based) Dependent counter).}} 
A temporary variable $dc$ of program $p$ is a (\beforestepdistance,\afterstepdistance)-(step-based) dependent counter iff for any computation $\rho = s_0, s_1, s_2,\cdots$ of $p$ the following condition holds: $\forall a : a \geq 0: $

\begin{enumerate}
\noindent\item $dc(s_a) = \bot \wedge dc({s_{a+1}}) \neq \bot \\
\Rightarrow \exists w : a-\beforestepdistance \leq w \leq a +1 :  dc({s_{a+1}}) = fc(s_{w})$,\\ where $fc$ is a free counter in $p$
\item $dc(s_a) \neq \bot \Rightarrow \forall w : w > a+ \afterstepdistance : dc(s_{w}) = \bot$
\item $dc(s_a) \neq \bot \wedge dc(s_{a+1}) \neq \bot \ \ \ \Rightarrow \ \ \ dc(s_{a}) = dc(s_{a+1})$

%\indent\indent\indent\indent\indent$dc(s_l) = f(fc(s_l)) \text{ or } f(dc_1(s_l))$
\end{enumerate}
\end{definition}

\begin{comment}
Note that when the values of $\beforestepdistance$ and $\afterstepdistance$ are clear from the context, we omit them and denote it just as dependent counter. 
%
Also, the above definition characterizes dependent counters in terms of number of program steps. We extend this definition in Section \ref{sec:sysmodel} (cf. Definition \ref{def:dep_counterregion}) by considering the execution time. 

\end{comment}

Recall that one of the assumptions in Section \ref{sec:sysmodel} was intended to translate the \textit{steps of a program} into the \textit{corresponding time.}
Based on this assumption, next we define the notion of (region-based) dependent counters where the value of the dependent counter is based on the value of free counters in preceding regions. 
%This assumption will play a role in how dependent counters are handled. A $(\beforestepdistance, \afterstepdistance)$ step based dependent counter had the property that (1) when the dependent counter is set to a value different from $\bot$, it is set to the value of some free counter in at most  $\beforestepdistance$ steps back, and (2) after the value of the dependent counter is set to a value different from $\bot$, within $\afterstepdistance$ steps it is set to $\bot$. 
In particular, we translate $\beforestepdistance$ and $\afterstepdistance$ in Definition \ref{def:dep_counter} into corresponding region values. We treat a counter as $(\beforeregiondistance, \afterregiondistance)$-dependent counter (1) when the dependent counter is set to a value different from $\bot$, it is set to the value of some free counter in at most $\beforeregiondistance$ \textbf{(global) regions} in the past, and (2) after the value of the dependent counter is set to a value different from $\bot$, within $\afterregiondistance$ \textbf{(global) regions} it is set back to $\bot$. Hence, we define region based dependent counters as follows:

\begin{definition}\label{def:dep_counterregion}{\bf{((Region based) Dependent counter).}} 
A temporary variable $dc$ of program $p$ is a (\beforeregiondistance,\afterregiondistance)-dependent counter iff for any computation $\rho = s_0, s_1, s_2,\cdots$ of program $p$ the following conditions hold: $\forall a : a \geq 0: $
\begin{enumerate}
\noindent\item$dc(s_a)=\bot$ $\wedge$ $ dc({s_{a+1}}) \neq \bot$ $\wedge$ $\{s_{a+1}$ is in (global) region $r$\}
$\Rightarrow \exists w :  w \leq a +1 :  dc({s_{a+1}}) = fc(s_{w})$,\\ where $fc$ is a free counter in $p$ and $s_w$ is in (global) region [$r - \beforeregiondistance$..$r$]
\item $dc(s_a) \neq \bot$\\ 
$ \Rightarrow \forall w :$ region of $s_w$ is greater than $r+\afterregiondistance : dc(s_{w}) = \bot$
\item $dc(s_a) \neq \bot \wedge dc(s_{a+1}) \neq \bot \ \ \ \Rightarrow \ \ \ dc(s_{a}) = dc(s_{a+1})$
\end{enumerate}

\end{definition}

%\fbox{\rule{0mm}{0mm}
%\begin{minipage}{3in}
\begin{remark}%{\bf Remark. } \ 
Observe that the above definition overloads the definition of step-based dependent counter. Specifically, we use the term $(\beforestepdistance, \afterstepdistance)$-(step-based) dependent counter while viewing the counter in terms of number of steps. And, we use $(\beforeregiondistance, \afterregiondistance)$ while viewing it in terms of regions. In the rest of the paper, unless specified otherwise, {\bf we assume that dependent counters are specified in terms of regions.}
\end{remark}

\remark{In a given system, irrespective of what kind of collection  (a set or a list or a sequence) that a complex variable may correspond to, our algorithm focuses only on bounding each constituent simple variable or entry in the complex variable, whereas the overall structure or the complex variable itself remains unaffected by the algorithm. 
In other words, operations associated with the data structure itself (e.g., the next element in the list) are performed as is. However, any operation on the data item (e.g., if first item in the list is equal to $0$) would be affected by our transformation algorithm. In this case, before the equality operation is performed, we apply the transformation based on the properties (defined in the subsequent discussion) of that list item. 

%Though the size of a complex variable itself is unaffected by our transformation algorithm, its size can be bounded by eliminating the constituent entries as soon as they become obsolete.}
%\end{minipage}
%}

\begin{comment}
\begin{remark}
Since the global region may differ from region of a process, our algorithm will account for possible discrepancies between them. We \textit{do not use regions associated with processes in the above definition} since dependent counters may be accessed by multiple processes. For example, $cl.m$, clock of message $m$ in Lamport's clock, is accessed by both the sender and the receiver. 
%Finally, in the rest of the paper, unless specified otherwise, {\bf all dependent counters are region based dependent counters.}

\end{remark}
\end{comment}

% * <sandeepkulkarnimsu@gmail.com> 2016-11-18T17:16:41.293Z:
%
% Please check the new definition
%
% ^.
\iffalse
\begin{definition}\label{def:dep_counter}{\bf{(Dependent counter).}} 
A temporary variable (or an entry of a complex variable) $dc$ of program $p$ with limited lifetime say $v$ is a dependent counter iff it is of the form : 

\indent\indent\indent\indent\indent$dc(s_l) = f(fc(s_l)) \text{ or } f(dc_1(s_l))$
\end{definition}

Here the function $f$ stands for a linear function. Briefly, the value of the dependent counter $dc$ of a program $p$ at a state $s_l$ depends on the value of some free counter $fc$ or some other dependent counter $dc_1$ of the program $p$.
\fi 
%

\bibliographystyle{plain}
%ILLUSTRATING FREE AND DEPENDENT COUNTERS
\label{sec:Examples}
%\section{Examples}

\section{Illustrating Free and Dependent Counters} 
\label{sec:lamport}

%The logical clock value of a process is free counter whereas the timestamp value of any event or message is a dependent counter since they are assigned to the logical clock value of the respective process.
In this section, we illustrate our definitions of free and dependent counters with the help of Lamport's logical clocks \cite{lamport}.
In this program, the processes in the system communicate through messages. 
At any point in time each process $j$ has a logical clock value $cl.j$ associated with it, and $cl.j$ increases whenever an event occurs at $j$.

Next, using our formalism from Section \ref{sec:modeldistributed}, we specify the actions of this program. Also, we identify the notion of simple versus complex variables, dependent versus free counters, etc. The actions of a process, say $j$, in this program are as follows:
\begin{enumerate}
\item Action Local Event

$true \longrightarrow cl.j = cl.j+d;$ %cl.a = cl.j$
\item Action Send Event, say to process $k$

$true \longrightarrow cl.j = cl.j+d; cl.m = cl.j;$

$channel_{j, k} = channel_{j,k} \cup \{m\}$.

\item Action Receive Event, say from process $k$

$m \in channel_{j,k}  \longrightarrow cl.j = max(cl.j, cl.m)+d;$

$channel_{j, k} = channel_{j,k} - \{m\} $
\end{enumerate}

where $d$ is any positive integer that can be different at different instances of the actions. 
%
%In this program, we have one variable $cl$ for each process. 
%
Observe that for every process $j$, $cl.j$ is a permanent variable. The variable $channel_{j,k}$ is a {\em complex} variable which contains timestamps of messages in transit. If we unravel this variable, we get multiple timestamps, each corresponding to a message in transit.

\begin{theorem}
%\footnote{Proof in Appendix \ref{subsec:proof_lamport}}
\label{thm:lamportcounter}
In Lamport's logical clock program, 
%\begin{itemize}
%\item $cl.j$ is a free counter, and 
%\item each entry in $channel_{j,k}$ is a (0,\maxsteps)-(step-based) dependent counter provided any message is guaranteed to be received within $\maxsteps$ steps. %\footnote{Add stepbased to other parts}
%\end{itemize}
%\end{theorem}

%\subsection{Proofs from Section \ref{sec:lamport}}
\label{subsec:proof_lamport}

%{\bf Proof of Theorem \ref{thm:lamportcounter}. }

\begin{enumerate}
\item $cl.j$ is a free counter

{\bf Proof.}
The permanent variable $cl.j$ is a free counter of process $j$ that satisfies definition \ref{def:free_counter}. In particular if $\langle s_0,s_1,s_2,\cdots,s_n\rangle$ is a computation prefix of $p$, then:

(i) at a given state $s_{l}$ when an event occurs, the value of $cl.j$ is computed as $cl.j(s_{l}) = cl.j(s_{l-1})+d$, $d > 0$, or $cl.j(s_{l}) = max(cl.j(s_{l-1}),cl.m) + d$, $d > 0$, i.e., it is higher than the logical clock value of $j$ in its previous state $s_{l-1}$. 
Thus $cl.j$ is an unbounded counter that has the form $cl.j(s_{l}) > cl.j(s_{l-1})$ i.e. it never decreases.
		
(ii) Also, if $\rho = \br{s_0, s_1, s_2,\cdots,s_l}$ is a valid computation prefix of $p$, appending state $s_{l+1}$ to $\rho$, where $cl.j(s_{l+1}) = cl.j(s_l)+ d$ results in $\rho' = \rho + s_{l+1}$ which is also a valid computation prefix that contains one extra event. 
In other words, an increase in the logical clock value of a process by $d$ continues to preserve the correctness of the overall system. 
Thus the logical clock value associated with any process is a \emph{free counter}.

 \item Each entry in $channel_{j,k}$ is a (0,\maxsteps)-(step-based) dependent counter provided any message is guaranteed to be received within $\maxsteps$ steps. 

{\bf Proof.}
Entries in $channel_{j,k}$ i.e., message timestamps are \emph{dependent counters} in the system, since they are temporary variables that have the form outlined in definition \ref{def:dep_counter}. In particular, let $\langle s_0,s_1,s_2,\cdots \rangle$ be a computation of $p$ and let $cl.m$ denote the timestamp of a message $m$ in $channel_{j,k}$ . Then, $cl.m(s_j)$ is equal to $\bot$ when $m$ is not in transit (before transmission or after reception) and $cl.m$ equals the timestamp of $m$ when message $m$ is in transit.

(i) if there exists a state $s_a$ such that $cl.m(s_a) = \bot$ and $cl.m(s_{a+1}) \neq \bot$ then this corresponds to sending of $m$. In this case, $cl.m(s_{a+1})$ is set to $cl.j(s_{a+1})$. It follows that this satisfies condition $1$ of Definition \ref{def:dep_counter}.

%in a given state $s_a$, if a message $m$ is not in $channel_{j,k}$ the value of $cl.m = \bot$. In the next state $s_{a+1}$ when message $m$ is added to the channel i.e. $m \in channel_{j,k}$, then $cl.m \neq \bot$ since $m$ is timestamped as $cl.m = cl.j$ before it is added to $channel_{j,k}$. In other words, the value of $cl.m$ is equal the value of some free counter in state $s_{a+1}$.

%i.e. it is of the form $dc({s_{a+1}}) = fc({s_{a+1})$, where $cl.j$ is the free counter f$c({s_a})$ and $k=0$ in this case.

(ii) if $cl.m(s_a) \neq \bot$ in some state $s_a$ then it means that message $m$ is in transit in state $s_a$. Since we assume that every message would be delivered in \maxsteps steps, it follows that after $\maxsteps$ steps, in state $s_{a+\maxsteps}$, message $m$ will no longer be in transit. It follows that this satisfies condition $2$ of Definition \ref{def:dep_counter}.

(iii) a message $m$ is timestamped only once, i.e. when it is added to $channel_{j,k}$ and $cl.m$ is set to $\bot$ only after it is removed from $channel_{j,k}$. When $m$ is in transmission the value of $cl.m$ is never changed. This satisfies condition $3$ of Definition \ref{def:dep_counter}.

%Thus in a given state $s_a$, if $cl.m \neq \bot$ then it means that $m \in channel_{j,k}$, and in the next state $s_{a+1}$, if $m$ is still in the channel i.e. $cl.m \neq \bot$, then its value is guaranteed to be unchanged. So it is of the form, $cl.m(s_a) \neq \bot \wedge cl.m(s_{a+1}) \neq \bot \ \ \ \Rightarrow \ \ \ cl.m(s_{a}) = cl.m(s_{a+1})$, since $s_a$ and $s_{a+1}$ would correspond to program states when $m$ is in transit.

\end{enumerate}

\end{theorem}

\section{Transformation Algorithm}
\label{sec:algorithm}

Our transformation focuses on a three-step approach. In the first step (Section \ref{sec:stepone}), we focus on revising a given program such that the free counters in the program, while still being unbounded, are closely related to the physical time. In the second step (Section \ref{sec:steptwo}), we do the same for dependent counters. Finally, in the third step (Section \ref{sec:stepthree}), we revise the program obtained in the second step such that all counters become bounded. 
Due to reasons of space, we illustrate our algorithm in the context of the example in Section \ref{sec:lamport} in Appendix \ref{sec:algexampleone}.

%Our algorithm relies on the assumption about how fast the (free) counters can grow in a given region of time. 

%\footnote{do we need this paragraph??}Recall that the region is obtained by partitioning physical time into regions. To provide uniform partitioning of regions, we consider the existence of a global clock (that is not known to processes) to identify the regions.
%However, we assume that $\regionsize$ is chosen in such a way that the pairwise difference between the global region and the region of any process is at most $1$. 
%Each process also has its own local clock that is kept {\em close} to this time using protocols such as NTP. We choose the region size such that (1) the region identified by the process from its own local clock differs from the region identified by the global clock by at most $1$, and (2) the regions identified by two processes from their local clocks differ by at most $1$. Clearly, if we choose region to be say $1$ hour, this would be trivial. 
%Practically, we can choose region to be a few milliseconds given the accuracy of NTP clocks. 

%%Figure 
\begin{figure}
  \hfill\includegraphics[width=80mm,height=27mm, scale=1.0]{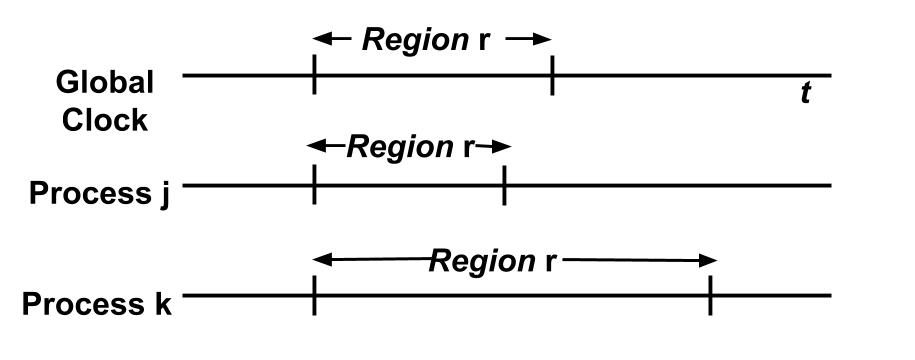}\hspace*{\fill}
  \caption{Regions with respect to global clock and processes}
  \label{fig:globalregion_vs_processregion}
\end{figure}

Our algorithm utilizes the observation that while the counters used in a program can grow unbounded, their growth in a given time period (assuming no transient faults) can be computed. In particular, consider a computation within one region as determined by the global clock. We assume that the growth of the counter (from its original value) would be bounded by a constant in this region. As an illustration, for the program in Section \ref{sec:lamport}, we can identify this bound by considering the number of events that could be created in the given region. Note that the region from the perspective of the global clock may not be the same as that of a process, (cf. Figure \ref{fig:globalregion_vs_processregion}). Hence, from the perspective of the process, the growth of the counter in {\em its} region may be different.

\subsection{Algorithm for Step 1: Adjusting Free Counters}
\label{sec:stepone}

Let $\maxinc$ be the maximum increase in any free counter in one global region.
Since free counters can be increased at will, the natural approach is to try to keep the value of the free counter in region $r$ to be $[r*\maxinc .. (r+1)*\maxinc]$. It turns out that this is not feasible since the process regions may not be identical. So, a process in region $r+1$ may send a message to process in region $r$ causing it to receive values that are outside this range. Hence, in our algorithm, we proceed as follows: we \textit{try to ensure that}
%
%(Observe that in Section \ref{sec:algexampleone}, the value of \maxinc was assumed to be $10$.) Based on our discussion in Section \ref{sec:algexampleone}, we will {\em try to} ensure that in \regionrr, the value of
any free counter is in the range $[3r*\maxinc .. 3(r+1)*\maxinc-1]$. However, in practice, since the regions of two processes may not be identical, we will ensure that the value of the free counter is in the range $[3r*\maxinc .. 3(r+1)*\maxinc+2*\maxinc-1]$.\footnote{The reason for parameters 2 and 3 in this equation is discussed in Appendix \ref{sec:algexampleone}, where we illustrate the first step with the problem of logical clocks.} Each process will first ensure that this constraint is satisfied. If it is not, it will restore the value of the free counter to $3r*\maxinc$, where $r$ is {\em its} current region. This can be achieved by checking the values of the free counter (1) as soon as the region of the process changes (II.1 in Figure \ref{fig:algo}), or (2) the process updates its free counter as part of its actions (III.4 in Figure \ref{fig:algo}), or (3) the process uses the free counter (in evaluating guard of an action) (III.1 in Figure \ref{fig:algo}). Thus, the algorithm for transformation is as shown in Figure \ref{fig:algo}.

\iffalse
\> {\bf For all actions $guard \longrightarrow statement$}\\
%\>\> Before evaluating guard\\
\>\> // The transformed program maintains\\ 
free/dependent counters in a modulo form.\\ 
First concert it to corresponding integer format.\\
\>\> For each free counter $fc$, obtain $intfc$\\
\>\>\> intfc = convert(fc)\\
\>\> Evalauate guards from the values of the converted couters\\

%\>\> Update the values of each counter based on current time\\
%\>\> (as described in Step 3)\\
\ \\
\>\> Evaluate the guards to identify one guard \\
\>\> that evaluates to true\\
\>\> Select one guard which evaluates to true and execute corresponding statement\\
\>\> While executing the statement \\
\>\>\> After any update to free counter $intfc$\\
\>\>\>\\
\>\>\>\> check(fc, r)\\
\>\>\> After any update to dependent counter $dc$\\

\fi

\begin{figure}[t]
{
\footnotesize
\begin{tabbing}
\hspace*{1mm} \= \hspace*{5mm} \= \hspace*{5mm} \= \hspace*{5mm} \= \hspace*{5mm} \= \hspace*{5mm} \= \hspace*{5mm} \= \hspace*{5mm} \= \hspace*{5mm} \= \kill\\
%\>\bf{ALGORITHM FOR STEP 1 AND STEP 2:}\\
\>{\bf I. Variables:}\\
\>\>$\maxinc : $ maximum increase in any free counter in one\\ 
\>\>global region\\
\>\>$r.j : $ region of process $j$ determined from its local clock\\
\>\>$fc.j : $ free counter of process $j$\\
\>\> $dc.j : $ \br{\beforeregiondistance,\afterregiondistance} dependent counter of process $j$\\
\>\>\> // The algorithm below is repeated for each free and\\ 
\>\>\> // dependent counter\\
\>\>$\maxv :$ maximum ($r_b+r_f$) value for dependent counters\\
\>\> $MAXBOUND = \maxbound$\\
\>\> $minfree = \minfreecounter{r}$\\
\>\> $maxfree = \maxfreecounter{r}$\\
\>\> $mindep = \mindepcounter{r}$\\
\>\> $maxdep = \maxdepcounter{r}$\\

\ \\
\>{\bf II. Whenever region associated with $j$ changes:}\\
\> \> 1. For each free counter $fc.j$ in $j$:\\
\>\>\>checkfc($fc.j, r$)\\
\> \> 2. For each dependent counter $dc.j$ in $j$:\\
\>\>\>checkdc($dc.j, r$)\\
\ \\
\> {\bf III. For each action $guard \longrightarrow statement$ of original program:}\\
%\>\> Before evaluating guard\\
\>\>\> // {\em transformed program maintains free/dependent}\\
\>\>\> // {\em counters in a modulo form. First convert it to}\\
\>\>\> // {\em corresponding integer format.} \\
\>\>1. \> For each free counter $fc$ in the $guard$:\\
\>\>\> $intfc = convertfc(fc)$\\
\>\>2. \> For each dependent counter $dc$ in the $guard$:\\
\>\>\> $intdc = convertdc(dc)$\\

\>\>3. \>  Evaluate $guard$s with updated counters and\\
\>\>\> select an action to execute\\
\>\>\> // Note: original program actions utilize unbounded counters.\\
\>\> 4. \> Whenever $intfc$  is updated in the statement\\
\>\>\>\> checkfc($intfc, r$)\\
\>\>\>\> set $fc = fc \ mod \ MAXBOUND$\\
\>\>5. \>  Whenever $intdc$  is updated in the statement\\
\>\>\>\> checkdc($intdc, r$)\\
\>\>\>\> set $dc = dc \ mod \ MAXBOUND$\\

\iffalse
\>\>\> if $guard$ evaluates to $true$ then:\\ 
\>\>\>\hspace*{5mm} \= Execute corresponding $statement$ \hspace*{5mm}\\
%\>\>\>\>While executing the $statement$:\\
\>\>\>\>\hspace*{5mm} \=If any update to free counter $fc$ then:\\
\>\>\>\>\>\hspace*{5mm} \=check(fc, r)\\
\>\>\>\>\> If any update to dependent counter $dc$ then:\\
\>\>\>\>\>\> checkdc(dc, r)\\
\fi
\ \\
\>{\textbf{IV. Function checkfc$( fc, r)$:}}\\
\>\>If ( ($fc < minfree$) $||$\\
\>\>\>\hspace*{2mm}\=($fc > maxfree$) ) then:\\
\>\>\>\>\>set $fc := minfree$ \\
\>\>\>\>\>//{\em reset free counter to minimum value}\\ 
\>\>\>\>\>{\em in the legitimate range}\\    
\ \\ 
\>{\textbf{V. Function checkdc$( dc, r)$:}}\\
\>\>If ($(dc < mindep)$ $||$\\ 
\>\>\>($dc > maxdep$) then:\\
\>\>\>\>\>set $dc := mindep$\\ 
\>\>\>\>\>//{\em reset dependent counter to minimum}\\ 
\>\>\>\>\>{\em value in the legitimate range}\\
\ \\
\>{\textbf{VI. Function convertfc$(x)$:}}\\
\>\> Find $y$ such that $x = y \ mod \ MAXBOUND$ and\\ 
\>\> $y$ is in range [minfree .. maxfree]\\
\>\> If no such $y$ exists, return minfree\\
\ \\
\>{\textbf{VII. Function convertdc$(x)$:}}\\
\>\> Find $y$ such that $x = y \ mod \ MAXBOUND$ and\\ 
\>\> $y$ is in range [mindep .. maxdep]\\
\>\> If no such $y$ exists, return mindep\\
\end{tabbing}
\vspace*{-3mm}
\caption{Our Transformation Algorithm}
\label{fig:algo}
}
\vspace*{-7mm}
\end{figure}

\subsection{Algorithm for Step 2: Adjusting Dependent Counters}
\label{sec:steptwo}

%Step 1 of the algorithm focused on changing the free counters in the given program $p$. Step 2 focuses on analyzing the role of dependent counters.

Let $dc$ be a  $(\beforeregiondistance, \afterregiondistance)$ region-based dependent counter. 
%In other words, when the value of $dc$ is set to a value other than $\bot$, it is set to the value of some free counter in preceding $\beforeregiondistance$ regions. And, after the value of $dc$ is set to a value other than $\bot$, within $\afterregiondistance$ regions, it is set back to $\bot$. 
Now, we identify the possible values of $dc$ that may happen under legitimate states, i.e., in the absence of faults. 

Consider the case where a process is in \regionrr, the value of its free counter is in the range \freelegbound{r}. At this time, the global region is at least $r-1$. If the counter $dc$ is used in global region $r-1$, then it was initialized in global region greater than or equal to $r-1-\afterregiondistance$. Moreover, the value it was set to can only come from a free counter $\beforeregiondistance$ regions earlier. In other words, the value of $dc$ was set to the value of a free counter in global region $r-1-\beforeregiondistance-\afterregiondistance$ or higher. Since the process region and global region may differ by $1$, the region of the process that set the value of $dc$ is at least $r-2-\beforeregiondistance-\afterregiondistance$. Hence, the value of $dc$ is at least \minfreecounter{(r-2-\beforeregiondistance-\afterregiondistance)} Moreover, the maximum value of $dc$ is the maximum value of some free counter, i.e., it is \maxfreecounter{r}.

Hence, in Step 2 of our algorithm, we ensure that the value of a given dependent counter is always within this range. If it is not in this range, we set it to the minimum permitted value in this range, i.e., we set it to \minfreecounter{(r-2-\beforeregiondistance-\afterregiondistance)}. Similar to free counters, this is done (1) as soon as the region of the process changes (II.2 in Figure \ref{fig:algo}), or (2) when the process sets a dependent counter as part of its actions (III.5 in Figure \ref{fig:algo}), or (3) the process uses the dependent counter (in evaluating guard of an action) (III.2 in Figure \ref{fig:algo}). 
Each dependent counter is characterized by parameters $\beforeregiondistance$ and $\afterregiondistance$. Let the maximum value of $\beforeregiondistance+\afterregiondistance$ for any dependent counter be $\maxv$, Thus, if a process is in \regionrr, its dependent counters must be in $[\minfreecounter{(r-2-\maxv)}..\maxfreecounter{r}]$.

\subsection{Algorithm for Step 3: Bounding the Counters}
\label{sec:stepthree}

Steps 1 and 2 focused on relating free and dependent counters to physical time. Recall that if a process is in \regionrr, then any free counter is in the range $[\minfreecounter{r}..\maxfreecounter{r}]$.  
Recall that the value of any dependent counter is in the range [\mindepcounter{r}..\maxdepcounter{r}]. 
Observe that the size of the above range is $\maxrange$. In Step 3, we revise the program so that instead of maintaining each counter to be an unbounded variable, we only maintain it in modulo $m$ arithmetic, where $m$ is 3 times the range of any dependent counter. In other words, $m= \maxbound$.

%\maxbound$ arithmetic. (Note that this number is three times the range of the dependent counters.) In other words, instead of maintaining a counter $c$, we only maintain $c$ $mod$ $\maxbound$. \SK{Removed -1 from maxbound definition}

%%Figure 
\begin{figure}
  \hfill\includegraphics[width=80mm, scale=0.7]{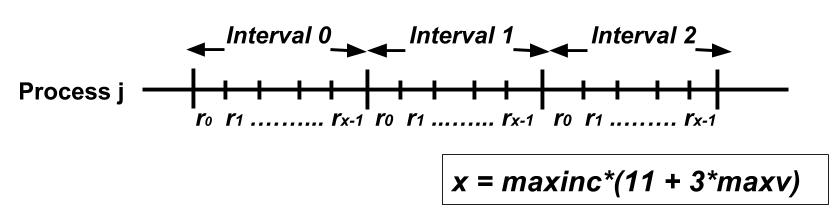}\hspace*{\fill}
  \caption{Determining bounds from Intervals}
  \label{fig:bounded_intervals_for_counters}
\end{figure}

Next, we give a brief description of why the value \maxbound is chosen. Towards this end, we split \maxbound into three {\em intervals},  $[0..\intervallength-1],[\intervallength..2\intervallength-1]$ and $[2\intervallength..3\intervallength-1]$. \textbf{Each interval corresponds to the range of dependent counters.}

%consists of \intervallength regions as shown in Figure \ref{fig:bounded_intervals_for_counters}. \SK{Changed intervallenght here check}

First observe that the interval is long enough to ensure that all free counters stabilize to their expected values. 
%As discussed in Appendix \ref{sec:algexampleone}, this takes at most 3 regions. 
%\ref{fig:bounded_interval regions. s_for_counters}. 
%
Regarding dependent counters, consider the case where a process, say $j$, is about to move from Interval $0$ to Interval $1$. Since the program can be perturbed to an arbitrary state, at this point, a dependent counter could be in any interval. However, any dependent counter that exists when $j$ is about to move to Interval $1$ will be removed from the system before process $j$ moves to Interval $2$. (Note that the length of the interval was chosen in order to guarantee this property.)
	
Now, consider the computation of the program where process $j$ just enters Interval $1$ and continues its execution until it enters Interval 2. During this computation, process $j$ will discard all dependent counters in Interval $2$. This is due to the fact that only valid values for dependent counters in Interval $1$ are from Interval $0$ or Interval $1$. 
%Furthermore, any new dependent counter values generated during Interval $1$ are either in Interval $0$ or Interval $1$. 
Moreover, given the life-span of dependent counters, any dependent counter generated in Interval $0$ will be removed before $j$ enters Interval $2$.
In other words, when the first process enters Interval $2$, all dependent counters are from Interval $1$. Moreover, this property will be preserved for all subsequent intervals. 

\subsection{Correctness of Transformation Algorithm}
%Let $p$ be the given stabilizing program with unbounded (free or dependent) counters. Let $p'$ be the transformed program based on the Algorithm in Figure \ref{fig:??}. Then, we need to show that if $p$ is stabilizing fault-tolerant then $p$' is stabilizing fault-tolerant as well. 

%{\bf Proof of correctness. } \
Let $p$ be the given stabilizing program. Let $p'$ be the program obtained after Steps 1, 2 and 3. Starting from a legitimate state, we show that there is a mapping of a computation of $p'$ to a computation of $p$. And, we also show that starting from an arbitrary state, any computation of $p'$ has a suffix that maps to a computation of $p$. Taken together, this shows that if $p$ is stabilizing then so is $p'$. For reasons of space, we present the proof in Appendix \ref{sec:ProofofCorrectness} in Theorem \ref{thm:correct}.

\iffalse
If $p$ is stabilizing to state predicate $S$ then $p'$ is stabilizing to $S^m$, where $S^m = \{ s^m | s \in S$ and $s^m$ is the modulo $m$ state of $s$ \}, where $m = \maxbound$. 

% of correctness of our algorithm in the Appendix.
\fi

\begin{figure*}
\centering     %%% not \center
\subfigure[]{\label{fig:LogicalClocks}\includegraphics[width=55mm,height=51mm]{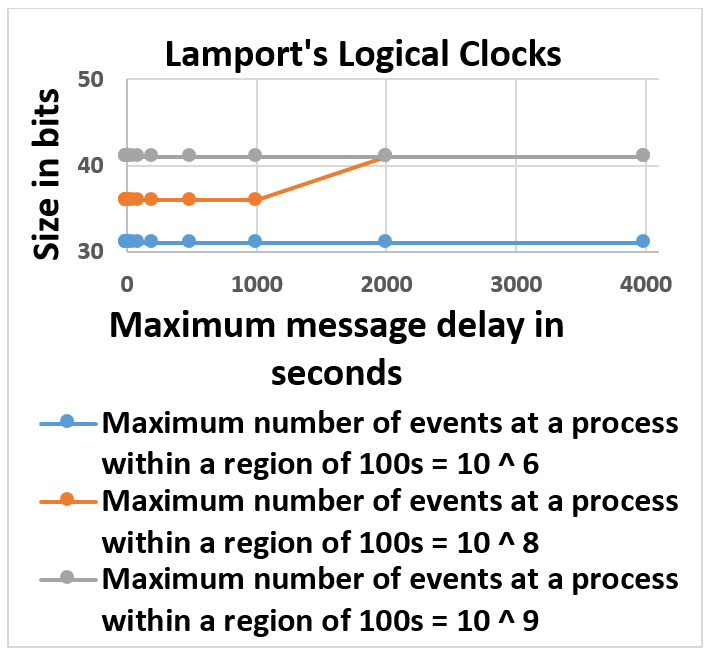}}
\subfigure[]{\label{fig:Katz_Perry}\includegraphics[width=55mm,height=51mm]{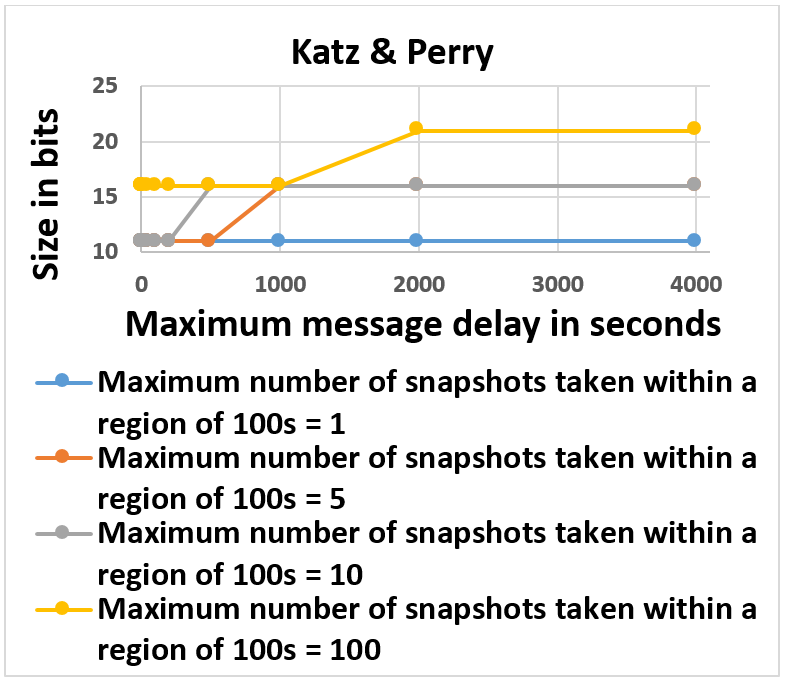}}
\subfigure[]{\label{fig:PaxosConsensus}\includegraphics[width=55mm,height=51mm]{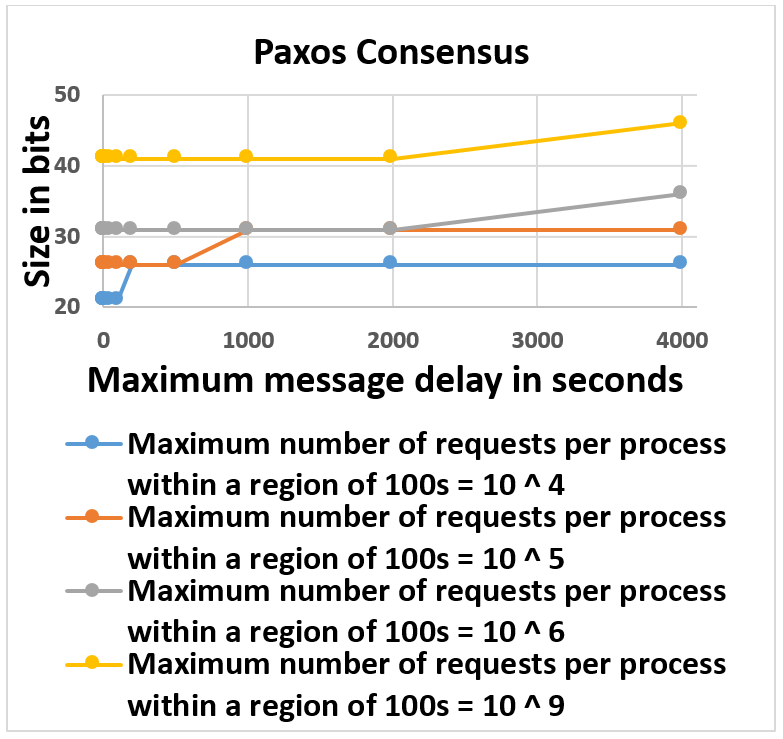}}
\caption{Analysis of required size in bits for implementation of (a) Lamport's Logical Clocks, (b) \cite{Katz1993}, (c) Paxos Consensus (acceptor time to respond = 1 second)}
\end{figure*}

%So at any point in time the value of counters $cl.j$ and $cl.m$ are computed as $cl.j$ $mod$ $\maxbound$ and $cl.m$ $mod$ $\maxbound$ respectively, where $\maxinc$ is $10$. Thus the values of the counters are bounded at $780$, so value of $cl.j=cl.j$ mod $780$ and $cl.m=cl.m$ $mod$ $780$.

%\input{algorithm.tex}

\section{Application of our Algorithm}
\label{sec:application}

In this section, we demonstrate how our algorithm can be used to transform stabilizing programs that use unbounded variables to stabilizing programs that use bounded variables and (practically bounded) physical time. First, in Section \ref{sec:katz}, we show that our approach can be applied to any algorithm that can benefit from earlier seminal work by Katz and Perry \cite{Katz1993} that focuses on adding stabilization to any program. One issue with that approach is that they need to utilize unbounded counters. We show that our approach can be used to bound counters in those applications. 
In Section \ref{sec:paxos}, we demonstrate its application in Paxos.
In Appendix \ref{sec:diff}, we show that our algorithm can be applied to diffusing computation which is also applicable to leader election, mutual exclusion, loop free routing, distributed reset, etc.
In Appendix \ref{sec:vc}, we show that our algorithm can be applied to vector clocks and the resulting algorithm is similar to that in \cite{rvc}. 
And, in Appendix \ref{sec:mutualexclusion}, we demonstrate that our approach can be used for mutual exclusion algorithm. 

Due to reasons of space, we  include only an outline of our approach and why the unbounded variables in these programs are either free or a dependent counter. We have chosen these applications because they demonstrate the generality of our approach and also they provide several insights into how the algorithm can be applied in a setting where free and dependent counters may not be immediately visible. These insights are discussed in remarks after corresponding sections. 

\subsection{Application in Katz and Perry Framework \cite{Katz1993} for adding Stabilization}
\label{sec:katz}

In  \cite{Katz1993}, Katz and Perry  presented an algorithm to add stabilization to an existing program. The key idea of the algorithm is as follows: (1) An initiator performs a snapshot of the system using the algorithm by Chandy and Lamport \cite{cl85}. (2) If the snapshot indicates that the program state is not legitimate then it performs a reset whereby the program state is restored to a legitimate state and the computation proceeds thereafter. 

To enable calculation of the snapshot and to perform reset without stopping the program execution, they utilize a round number --\textit{which is an unbounded integer variable}-- This value is incremented every time a new reset is performed. Furthermore, if a process in round $x$ receives a message with round $y$ then (1) if $x < y$, the process moves to round $y$, (2) if $x=y$ then the process treats it as a normal program execution and (3) if $x > y$ then it ignores that message. 

While the round number in \cite{Katz1993} does not meet the requirements of the free or dependent counter directly, we can modify the program slightly so that the new variables meet the constraints of free/dependent counter. First, we change the algorithm so that after every snapshot, a reset is performed. However, a Boolean variable is used to identify that this reset is fake and, hence, processes should simply move to the next round but continue accepting previous messages and not perform actual reset. 
Towards this end, we maintain the following variables: (1) variable $nr$ that is maintained at the initiator only to identify the next round it should use for performing reset, (2) variable $cr$ that is maintained at all processes to identify the current round. This variable is in the same spirit as the algorithm in \cite{Katz1993}. (3) variable $b$ is a Boolean that identifies whether the reset being done is real or fake. This variable is always $0$ except for the moment when the process wants to perform a {\em real} reset because the snapshot indicated that the state is not legitimate, and (4) variable $lr$ that denotes the sequence number of the last {\em real} reset performed in the system.

With this change, we can observe that (1) $nr$ is a free counter; it can be increased at will. (2) $cr$ can be viewed as a dependent counter provided we split the action that increases the current round into $cr = \bot;$ and $cr = $ {\tt new value of current round}. In this case, whenever $cr$ is changed, we treat it as if it was first set to $\bot$ thereby removing this dependent counter from the system. Then, we set it to the new value thereby creating a new dependent counter in the system. (3) $lr$ is a dependent counter since $lr$ needs to be set when a new {\em real} reset is performed. It can be set to $\bot$ after sufficient time to ensure that all existing messages in the system have been delivered. 

In \cite{Katz1993}, authors have utilized channel bounds as a mechanism to bound the counters. The above discussion shows that bounding is possible without using channel bounds.

Figure \ref{fig:Katz_Perry} shows the size of counters that is sufficient to preserve stabilization provided by \cite{Katz1993}. For the sake of analysis, we consider parameters that are satisfied by any \textit{practical} system. Hence, we consider the clock drift between any two processes as at most 100 seconds. Note that protocols such as NTP \cite{ntp,PrecSyncofCompNetClocks} provide clock synchronization within 100 milliseconds. In such a system, we consider different parameters for message delay on a single channel and number of resets performed in one 100 second window to compute the size of counters. Even if one makes really conservative assumptions namely, a message delay of up to one hour and as many as 100 resets could be performed in a 100 second window, the size of the counters is very small ($21$ bits). The size is even lower for more reasonable parameters. Furthermore, Figure \ref{fig:Katz_Perry} shows that the size of counters is not very sensitive to the message delay and number of resets in one window. Therefore, the designer can make very conservative assumptions without increasing the size of counters.

\begin{remark}%{\bf Remark. } \ 
Note that the above discussion also illustrates that neither the requirement that the dependent counter cannot be changed nor that it is reset to $\bot$ is restricting. Essentially, we need to treat an update of a dependent variable as a two-step process where we first remove the old value of the dependent counter and then initialize it as a new dependent counter (which happens to have the same name) with the new value. 
All that our definition requires is that the old value is no longer relevant for the subsequent computations and that the new value of the dependent counter is set to a recent value of some free counter. 
\end{remark}

\subsection{Paxos Based Consensus}
\label{sec:paxos}

A Paxos based consensus protocol has the following features: (1) Proposer $c$ proposes a {\em prepare} request with a sequence number $c.seq$ to the {\em acceptors} (2) Each replica {\em accepts} the request if it has not accepted a request with a higher sequence number. To do so, each acceptor $a$ maintains $a.seq$ which is the highest sequence number it has seen. (3) If an acceptor replies {\em NO}, it also notifies the proposer the value of $a.seq$ so that the proposer can choose a number higher than $a.seq$ for its subsequent request. (4) If a proposer receives sufficiently many {\em YES} responses (the precise number depends upon the number of failstop/byzantine faults we want to tolerate) it sends {\em accept} request to the {\em acceptors}. (5) An acceptor {\em accepts} this request iff it has not already responded to a prepare request with a higher sequence number. (6) A value is {\em chosen} provided sufficiently many acceptors {\em accept} the accept request. 

Observe that in this protocol we have sequence numbers maintained by proposers and acceptors. We associate two sequence numbers for each  proposer; $PendingSeq$ that denotes the sequence number of a pending request, if any. And, $NextSeq$, that denotes the sequence number it would use for a future request. 
Observe that $NextSeq$ is a free counter. The proposer can increase it at will without affecting the correctness of the Paxos based consensus algorithm. On the other hand, $PendingSeq$ is a dependent counter; it is set to be equal to $NextSeq$ whenever a request is made. As long as there exists a bound on message delivery and time required for acceptors to send a {\em YES} or {\em NO} message, $PendingSeq$ will be valid for a limited time since  each pending request will be accepted or rejected within a finite time. After this time, the value of $PendingReq$ will no longer be relevant and can be set to $\bot$. If the proposer chooses to send a new request, $PendingReq$ will be set to a different value. 

\begin{remark}
A paxos algorithm typically uses only one variable to model $NextSeq$ and $PendingSeq$, which is always an integer (and is never set to $\bot$). However, in this case, this variable is neither a free nor a dependent counter, as it is never set to $\bot$ and it cannot be increased at will. However, by having two variables, we can observe that $NextSeq$ is a free counter. To make $PendingReq$ a dependent counter, we split the action where the proposer learns that its previous request has failed (when we set $PendingReq$ to $\bot$) and when it starts a new request. In this case, each instance of the pending request is a new dependent counter. 
\end{remark}
Finally, the sequence number associated with acceptors is also a dependent counter. It is relevant only until (1) it receives a new request with a higher sequence number or (2) if it has not received a request for a long enough time, thereby it can treat a future request as if it is the first request it has ever received. 
%Thus, our algorithm can be applied to bound the variables of a stabilizing paxos based consensus algorithm. 

Once again, we use very conservative assumptions to identify the size of these counters. Similar to Section \ref{sec:katz}, we assume that clocks are synchronized to be within $100$ seconds of each other. In such a system, even if a message can be delayed upto 1 hour and there are $10^9$ requests in one $100$ second window, $46$ bits are sufficient. And, for more reasonable assumptions, even less bits are required for each counter. Once again, since the number of bits do not increase substantially as we increase message delay/number of requests, the designer can utilize extremely conservative assumptions. For example, the number of bits for a counter only increases from $41$
 to $46$ bits even if the message delay is increased from $1$ second to $4000$ seconds.

\section{Discussion and Related Work}
\label{sec:related}

One of the questions raised by our work is whether the timing properties utilized in our transformation algorithm affect the generality of the algorithm. We note that given the impossibility of solving consensus, leader election  and several other interesting problems in asynchronous systems \cite{flp}, any fault-tolerant solution to these programs must make some reasonable assumptions about the underlying system. Some typical guarantees are process speeds, message delays etc. Our algorithm utilizes assumptions of this nature to identify free and dependent counters. \textbf{Also, as shown in our case studies, even trivially satisfiable requirements --such as clocks differ by at most $100s$ (when current state of art guarantees synchronization to be less than $10$ milliseconds) or number of events in a given region is $10^9$ or a message is delivered within an hour-- are suffice to bound the variables within acceptable limits.  }

Not all programs that use unbounded counters can be used with our transformation algorithm. For example, consider algorithms such as those for causal broadcast that maintain an unbounded counter to keep track of the number of messages sent by each process. We cannot treat this as a free counter since incrementing it would require us to send broadcast messages. In other words, there are programs where unbounded counters may be neither free nor dependent. 

Our work also differs from previous work that uses distributed reset mechanism \cite{ag94,multireset} to bound the values of counters. Distributed reset affects all processes. By contrast, stabilization can often be achieved by only processes in the {\em vicinity} of the affected processes \cite{dgx07contain,gghp07contain}. Compared with the work in \cite{Blanchard2014} which assumes the counter size to be equal to the size of integers (32/64 bit in most systems), our approach has the potential to reduce the size of the counters. For example, the analysis from Section \ref{sec:algstepthreeexample}, shows a bound of $780$ is sufficient. In other words, the bound depends upon the need of the given application.  
%that counters are bounded at $780$ which is very small when compared to the $64$ bit counters 
%
%By contrast, in \cite{Blanchard2014}, the bound is decided by the size allocated to integers. 
%
Also, the algorithm in \cite{Blanchard2014} requires multiple/all processes to reset their counters if some process has to reset its counters. By contrast, our algorithm, when applied in the context of Paxos, addresses this issue by ignoring messages and resetting processes whose counters are affected rather than affecting all processes. Thus, if perturbation is small, it is anticipated that our solution will affect only the corrupted processes. 
%

%
%In \cite{Blanchard2014} every epoch change in tag forces the ballot numbers of the processes to be reset to zero, whereas with our algorithm when a counter is corrupted it is reset to the minimum value in its legitimate range and such a reset is performed only in the affected process.
%Also, in \cite{Blanchard2014} the processes are required to provide each other information about their local state (requests accepted so far) to avoid inconsistencies between them, whereas in our algorithm the use of regions (based on global time) eliminates the need for exchange of any local process related information, since processes avoid mutual inconsistencies implicitly by just staying within their legitimate regions.

\section{Conclusion and Future Work}
\label{sec:concl}

\begin{comment}
We presented an algorithm to transform a stabilizing program with unbounded variables into a corresponding stabilizing program that only utilizes bounded variables and physical time. Our algorithm relied on classifying unbounded variables into free counters and dependent counters. Intuitively, the former required that they can be increased at will whereas the latter required that they have a limited life. 
\end{comment}

Our work addresses a key conflict in the context of stabilization: (1) use of unbounded variables in stabilizing programs should be avoided since any implementation of that stabilizing program would rely on allocating large enough but bounded memory to each variable and transient faults could perturb the program to a state where the large bound associated with the variable is reached, and (2) use of (practically bounded) physical time is used in many systems because corruption associated with time is typically easily detectable and correctable. It provides an alternate approach for providing \textit{practically} bounded-space stabilization by utilizing system and application properties such as clock synchronization properties, message delivery properties, etc. Since a rich class of problems easily admit unbounded state-space solutions, our approach can be used to provide solutions where all program variables are bounded. 

We demonstrated that our algorithm is applicable in several classic problems in distributed computing, namely logical clocks, mutual exclusion, vector clocks, diffusing computation and Paxos based consensus.
We also demonstrated that our work can be combined with that of \cite{Katz1993}. This work transforms a given program into a stabilizing program with unbounded counters. Our work can be used to convert those unbounded counters into bounded counters while still preserving stabilization.
This work also demonstrates that for a rich class of programs, the approach taken by non-stabilizing programs to deal with unbounded variables --provide large enough but bounded space-- is feasible even with stabilizing programs.
\begin{comment}

In our work, we chose the size of the region so that region of any two processes differs by at most 1. We anticipate that by choosing a more fine grained value of region, the value of $\maxinc$, i.e., the value by which a dependent counter may increase in one region, it would be possible to reduce the size of the dependent counters. 
\end{comment}

%\newpage
\bibliography{sandeep,vidhya}
%\newpage
\appendix
\section{Proof of Correctness}
\label{sec:appendix}

\noindent In this section, we present proofs of correctness and the step-by-step illustration of our algorithm along with some of its applications --that were omitted due to reasons of space. 
\subsection{Illustration of Our Algorithm in Logical Clocks (Section \ref{sec:lamport})}
\label{sec:alg_illust}

\textbf{Illustration of the Step 1. }
\label{sec:algexampleone}

In this section, we identify the reason for choosing free counters of a process in region $r$ to be in the range 
[\minfreecounter{r}..\maxfreecounter{r}]. 

We use the context of logical clocks by Lamport from Section \ref{sec:lamport} to identify how this bound was derived. Note that in this program, the value of $cl.j$ for every process $j$ is a free counter.

For the sake of illustration, let us assume that the maximum increase in any free counter (i.e., $cl.j$ for any process $j$) in one global region is at most $10$. In other words, the value of $cl.j$ in a computation that goes on for $\regionsize$ global time increases by at most $10$. Assume that initially, all values of $cl.j$ are $-1$ (i.e.,logical clock value of every process) and the region of every process is $-1$. As soon as it creates the first event, it is in \regionzero. Thus, in one $\regionsize$ time, the value of $cl.j$ would increase to at most $10$. In $2\regionsize$ time, it would increase to at most $20$ and so on. 

Our first attempt to revise this program would be to require that in  \regionzero the value of $cl.j$ would be between $[0..9]$. In \regionone, $cl.j$ would be between $[10..19]$ and so on. Observe that the first property is already satisfied by the original program. However, the second property may be violated since $cl.j$ may be less than $10$ even in \regionone. We can remedy this by increasing the value of $cl.j$ as needed. Note that in this instance, the fact that $cl.j$ is a free counter is important, as it guarantees that $cl.j$ will never decrease and we are permitted to increase $cl.j$ as needed. With \regionone, we need to ensure that $cl.j$ does not increase beyond $19$, as we are not allowed to decrease it. We can try to ensure this property by the length of computation which guarantees a bound on the number of events that can be created in \regionone.

While the above approach is reasonable, it suffers from a  problem that the processes do not always agree on what the current region is. In particular, process $j$ could be in \regionone but process $k$ could still be in \regionzero. Now, if $j$ sends a message to $k$, it can cause $k$ to have a value for $cl.k$ that is outside $[0..9]$. 

Also, if process $j$ moves quickly to \regionone while process $k$ is still in \regionzero then it creates some additional difficulties. In such a system, the clock synchronization may force $j$ to {\em slow down} its clock to ensure that $k$ can catch up. (An alternative is to let process $k$ advance its clock more quickly. But we assume that we do not control the clock synchronization algorithm.)  In other words, as far as process $j$ is concerned, even if it starts with initial value of $cl.j$ to be $10$ at the beginning of \regionone, the value of $cl.j$ may exceed $19$ before $j$ enters \regiontwo. 
We can remedy the above problems with the observation that region of two processes differ by at most $1$. So, even if the  clock of $j$ is forced to slow down to let $k$ catch-up, as long as process $j$ is in the same region, its clock \textbf{will not increase by more than $30$}.  (Note that the value $30=10$*$3$ is due to the fact that \regionone of process $j$ can overlap with global regions $0, 1$ and $2$ and in each global region the increase in $cl.j$ is bounded by $10$.)

With this approach, we proceed as follows: In \regionzero, we {\em try to} ensure that the value of $cl.j$ is between $[0..29]$, in \regionone, the value of $cl.j$ is between $[30..59]$, in \regiontwo, the value of $cl.j$ is between $[60..89]$ and so on. Observe that with this change, when the first process, say $j$, moves to \regionone, all values were less than $30$. Based on the assumption about number of events in the region, as long as process $j$ is in \regionone, $cl.j$ cannot increase beyond $59$. 

We can summarize the above approach by the constraint that if the region of process $j$ equals $r$ then {\em we try to} ensure that $cl.j$ is between $[30r .. 30r+29]$. However, while process $j$ is in \regionrr, process $k$ could move to \regionrpone and if process $k$ communicates with process $j$, it could force process $j$ to increase $cl.j$ to be more than $30r+29$. However, as long as process $j$ is in \regionrr and process $k$ is in \regionrpone (which can last for at most $2\regionsize$ time), values of $cl.j$ or $cl.k$ cannot increase beyond $30r+29+20$.

Based on this observation, we define \minregion and \maxregion that identify the minimum region value held by some process and maximum region value held by some process. (Note that these values differ by at most $1$. Furthermore, the processes themselves are not aware of these values. They only know that the value of their region equals one of them.) From the above discussion, we observe that the value of $cl.j$ is in the interval $[30\minregion..30\minregion+29+20]$. Moreover, it is also guaranteed to be in $[30\maxregion-30..30\maxregion+19]$.
In addition, due to the property of regions, at some point, all processes must be in the same region. (If some processes are in region $r-1$ and some are in \regionrr, then no process can move to \regionrpone as long as some process is in region $r-1$. Thus, just before the first process moves to \regionrpone, all processes must be in the same region, namely \regionrr.) Let this region be $r$. Clearly, $\minregion$ and $\maxregion$ are equal to $r$ at this time. From the above discussion, at this point, $cl.j$ must be in $[30r .. 30r+19]$. 

The above analysis is correct if we assume that the value of $cl.j$ started with initialized values. However, if the values are corrupted, the above property may not hold. To rectify this, we change the value of $cl.j$ when we know that it is corrupted. For example, let process $j$ be in \regionrr. Since $\minregion$ is at least $r$, the value of $cl.j$ is at least $30r$. Also, since $\maxregion$ is at most $r+1$, the value of $cl.j$ is at most $30(r+1)+19$. If process $j$ finds itself in a situation where the value of $cl.j$ is outside this domain, then it resets it to $30r$. 

Additionally, as discussed above, there exists a time when all processes are in the same region. Given the local correction action to ensure that $cl.j$ is in the range $[30r .. 30(r+1)+19]$, it follows that when the first process is about to move to \regionrpone, the highest $cl$ value of any process is $30(r+1)+19$. When the first process moves to region $r+2$ (the overlap can be with at most two global regions), the increase can be up to $20$ so the highest $cl$ value is $30(r+2)+9$. Using the same argument, when the first process moves to region $r+3$, all $cl$ values are less than $30(r+3)$. Thus, when all processes move to region $r+3$, their $cl$ values are in the range $[30(r+3)..30(r+3)+29]$. And, this property is preserved for all future regions. Observe that this implies that within 3 regions, the value of the free counters is within their expected range. 

The above discussion rested on the assumption that the number of events created in one region is at most $10$.
Our algorithm generalizes this (as $\maxinc$) to adjust the free counters in the first step of our algorithm.

%%%%%%%%%%%%%%%%%%%%%%%%% ILLUSTRATION OF STEP 2 for LOGICAL CLOCKS
\textbf{Illustration of the Step 2.}
\label{sec:algexampletwo}

For the sake of illustration, suppose that in Lamport's algorithm for logical clocks, any message sent in region $x$ will be received or lost before region $x+5$. 

%(For example, if region is 1 hour long, this would correspond to a message being delivered in 5 hours or being lost. This could also be achieved by using wrappers that delete old messages.)

In Logical clocks, the value of $cl.m$ is a dependent counter. When the value of $cl.m$ is set, it is set to some current value of free counter (namely, $cl$ value of the sender process). Moreover, the value will be available for at most $5$ additional regions. In other words, $cl.m$ is a \br{0,5}-dependent counter. Hence, the above analysis requires that when a process receives a message $m$ in region $r$, it checks whether its timestamp is at least \minfreecounter{(r-7)} and at most \maxfreecounter{r}, where \maxinc equals $10$. 

%%%%%%%%%%%%%%%%%%%%%%%%% ILLUSTRATION OF STEP 3 for LOGICAL CLOCKS
\textbf{Illustration of the Step 3.}
\label{sec:algstepthreeexample}

Since we assume that $cl.m$ is a \br{0,5}-dependent counter, $\maxv=5$. Putting this value in \maxbound, we obtain 780. Hence, instead of maintaining variables $cl.j$ and $cl.m$, we maintain them as $cl.j \ mod \ 780$ and $cl.m \ mod \ 780$ respectively. Thus, 10 bits are sufficient to represent $cl.j$ and $cl.m$.

%Next, we generalize this to obtain the first step of our algorithm. 

\subsection{Proof of Correctness}
\label{sec:ProofofCorrectness}

In this section, we show that our algorithm consisting of 3 steps  preserves correctness and stabilization property of the original program. Let $p$ be the original program and let $p'$ be the program obtained after all 3 steps. To facilitate the proof we define the notion of {\em modulo state}.

\begin{definition}
\label{def:modstate}{\bf{(modulo $m$ state).}} 
Let $s$ be a state, modulo $m$ state of $s$, where $m$ is an integer, denoted by $s^m$ is obtained by changing each variable $x$ with unbounded domain in $s$ to $x$ $mod$ $m$. 
\end{definition}

We extend this to {\em modulo state predicate} and {\em modulo computation}. For example, $s_0^m, s_1^m \cdots$ is a modulo $m$ computation of $p$ iff $s_0, s_1, \cdots$ is a computation of $p$ and $\forall w: s_w^m$ is the modulo $m$ state of $s_w$. 

\begin{lemma}
\label{thm:closure}
In the absence of faults, (i.e., starting from a valid initial state), any computation of $p'$ is also a modulo $m$ computation of $p$, where $m = \maxbound$.
\end{lemma}

\begin{proof}
Observe that the bounds on free counters are derived based on the property that in any global region, the value of the free counter would be in the range $[\minfreecounter{r}..\maxfreecounter{r}]$. Hence, starting from an initial state, there would be no need to reset a free counter before/after execution of an action. Likewise, there is no need to reset a dependent counter. 
The only additional change to free counters is when a process moves from one region to another.

Now, consider a computation of $p'$, say $\br{s_0, s_1, s_2, \cdots}$. Since $s_0$ is an initial state of $p$, $\br{s_0}$ is a modulo $m$ computation-prefix of $p$ from an initial state.

Next, assume that \br{s_0, s_1, \cdots s_j} is a modulo computation-prefix of $p$. 
From the above discussion, $(s_j, s_{j+1})$ either executes an action of $p$ or it increases a free counter. In the former case, \br{s_0, s_1, \cdots s_j, s_{j+1}} is clearly a modulo $m$ computation of $p$. In the latter case, \br{s_0, s_1, \cdots s_j, s_{j+1}} is also a modulo $m$ computation of $p$ by the property of free counters, namely that their value can be incremented at anytime. 

From the above discussion it follows that $\br{s_0, s_1, s_2, \cdots}$ is a modulo $m$ computation of $p$, where $m = \maxbound$. 
\end{proof}

\begin{lemma}
\label{thm:conv}
A computation of $p'$ that starts from an arbitrary state has a suffix that is a modulo $m$ computation of $p$, where $m = \maxbound$.
\end{lemma}

\begin{proof}
Consider a computation of $p'$. Due to local correction of free counters, the value of any free counter of a process in \regionrr would be in the range $[\minfreecounter{r}..\maxfreecounter{r}]$. Note that even with taking modulo under \maxbound, this value can be uniquely identified given the physical time, 
As discussed above, we partition \maxbound into three intervals as shown in Figure \ref{fig:bounded_intervals_for_counters}. Wlog, let us assume that the current Interval is $0$ as determined by the physical time. As discussed above, if we consider the computation of $p'$ in the entire Interval $1$ then when the first process enters Interval $2$, all counters are within Interval $1$. Hence, the value of any counter can be uniquely determined from the physical time even if it is maintained in modulo \maxbound  arithmetic. 
Thus, from this point forward, every transition of $p'$ is also a modulo $m$ transition of $p$. 
It follows that, every computation of $p'$ has a suffix that is a modulo $m$ computation of $p$, where $m = \maxbound$.
%Hence, any subsequent computation is also a computation of the original program $p$. 
\end{proof}

\begin{theorem}
\label{thm:correct}
If $p$ is stabilizing to state predicate $S$ then $p'$ is stabilizing to $S^m$, where $S^m = \{ s^m | s \in S$ and $s^m$ is the modulo $m$ state of $s$ \}, where $m = \maxbound$. 
\end{theorem}

\begin{proof}
To show that $p'$ is stabilizing, we need to show that
\begin{enumerate}
\item If $p'$ starts from an arbitrary state then it recovers to its legitimate states. This follows from Lemma \ref{thm:conv}.

%This follows from Lemma \ref{thm:conv}. In particular, every computation of $p'$ has a suffix that is a modulo $m$ computation of $p$. And, every computation of $p$ is guaranteed to reach a legitimate state.
\item If $p'$ starts from a legitimate state then in the absence of faults, it remains in legitimate states forever and it is correct with respect to its specification. This follows from Lemma \ref{thm:closure}.
\end{enumerate}
\end{proof}
\subsection{Application of our Algorithm in Diffusing Computation}
\label{sec:diff}

\begin{figure*}
\centering     %%% not \center
\subfigure[]{\label{fig:DiffusingComputation}\includegraphics[width=53mm,height=50mm]{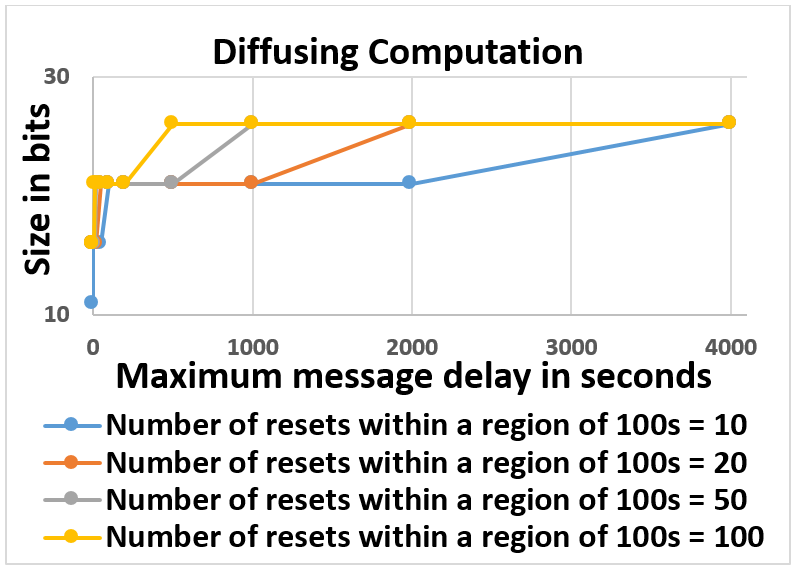}}
\subfigure[]{\label{fig:VectorClocks}\includegraphics[width=51mm,height=50mm]{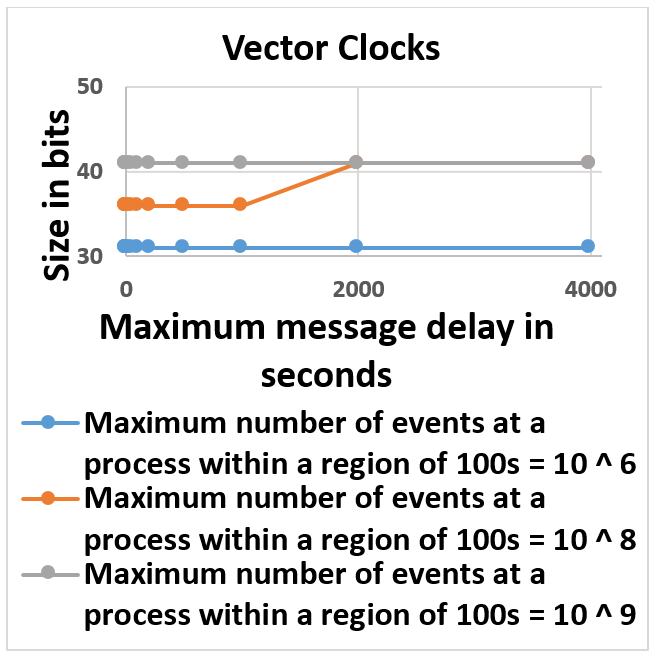}}
\subfigure[]{\label{fig:MutualExclusion}\includegraphics[width=51mm,height=50mm]{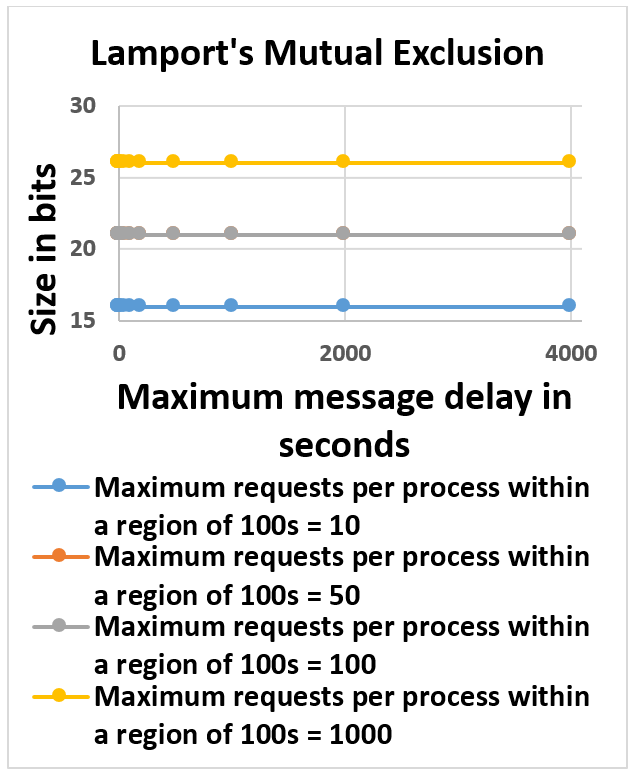}}
\caption{Analysis of required size in bits for implementation of (a) Diffusing Computation (number of processes = 100), (b) Vector Clocks (t=10), (c) Lamport's Logical Clocks based Mutual Exclusion (number of processes = 100).}
\end{figure*}

The problem of diffusing computation \cite{ds2} is intended to check/modify all processes in a given system. It is used in many applications such as ensuring loop free routing\cite{loopfree}, leader election \cite{vasudevan,Lee2007}, termination detection \cite{ds2}, mutual exclusion \cite{ak98masking}, and distributed reset \cite{ag94}. A key property of diffusing computation is that some process (or more than one process) may initiate it. This process is called the initiator.
Upon initiation, the process forwards it to its neighbors. This is called the {\em propagation phase}. When the neighbors receive this diffusing computation for the first time, they forward it to their neighbors. If they receive the same diffusing computation again --it can happen since there are several paths from the initiator to the given process-- they acknowledge it but do not forward it to others.
When a process receives the diffusing computation from all its neighbors, it begins the {\em completion phase} and sends an acknowledgement to its parent, i.e., the process from which it received the diffusing computation for the first time. 
When the initiator completes its diffusing computation, it is guaranteed that all processes in the system (that were present throughout the computation) have received and completed their diffusing computation.

One important requirement of diffusing computation is that a process will have to know whether the diffusing computation that it has received is the same as the one it had received before. This is achieved by using sequence numbers; the initiator utilizes a higher sequence number every time it begins a new diffusing computation. A straightforward approach to achieve this is via an unbounded sequence number. (If there are multiple initiators, we can use the ID of the initiator and sequence number.) 

Following our algorithm, we can observe the following:
\begin{itemize}
\item The sequence number of the initiator is a free counter; it can be increased by any value when the initiator begins a new diffusing computation. 
\item The sequence number at other (non-initiator) processes is a dependent counter. It is only relevant when the process begins propagation of the diffusing computation and ends when it completes the diffusing computation. The specific values of $(\beforeregiondistance, \afterregiondistance)$ for this dependent counter would be determined by the worst case time it would take for a diffusing computation to complete. 
\end{itemize}

Thus, our algorithm can be applied to bound the variables of a stabilizing diffusing computation algorithm. Figure \ref{fig:DiffusingComputation} shows the size of counters needed to achieve this. Once again, even if the number of diffusing computations initiated by \textit{one process} in a $100$ seconds window is $100$ and message delay is up to $1$ hour, the number of bits required is $26$. Furthermore, a process can ensure that the number of diffusing computations initiated by it satisfies this limit by simply counting the number of resets in one window. 

%\SK{Add references from iffalse below}
\iffalse

[23] S. Lee, M. Rahman, and C. Kim, “A Leader Election
Algorithm Within Candidates on Ad Hoc Mobile
Networks,” Embedded Software and Systems,
Lecture Notes in Computer Science, Vol. 4523, pp:
728-738, Springer Berlin / Heidelberg 2007

\fi

\subsection{Application of our Algorithm in Vector Clocks}
\label{sec:vc}

We discussed the application of our algorithm in Lamport's logical clocks in Section \ref{sec:lamport}. We can also extend it to vector clocks \cite{fidgeVC,matternVC} or hybrid vector clocks \cite{opodis2015}. Vector clocks maintain the variable $vc.j.k$ for each pair $j$ and $k$. This variable captures the knowledge that $j$ has about $k$. And,  $vc.j.j$ denotes a counter maintained by $j$ for itself. In this program, $vc.j.j$ is a free counter; $j$ can increment the counter maintained by itself by any value it desires. For $j \neq k$, $vc.j.k$ is a dependent counter, provided the underlying communication graph is strongly connected and there exists a time $t$ such that a message (timestamped with vector clocks) is sent on every link in time $t$. Following this approach, we can obtain a stabilizing program for vector clocks that uses bounded counters. The resulting program is the same as that in \cite{rvc}. In other words, our algorithm can be utilized to derive the program in \cite{rvc}. The size of a counter with vector clocks is small as shown in Figure \ref{fig:VectorClocks} even in scenarios where $10^9$ events are created in each window (of size 100 seconds) and message delay as long as an hour.

\subsection{Application of our Algorithm in Mutual Exclusion}
\label{sec:mutualexclusion}

The classic algorithm by Lamport \cite{lamport} for mutual exclusion utilizes logical clocks (recalled in Section \ref{sec:lamport}). It works as follows: (1) All messages are timestamped with logical timestamps presented in Section \ref{sec:lamport}. (2) When a process wants to access the critical section, it sends a request to all processes. (3) When a process receives the request, it adds the request timestamp to its queue and replies to the requesting process. (3) A process enters critical section iff it has received replies from all processes and if the smallest request contained in its queue corresponds to its own request. And, (4) finally, after a process is done with its critical section, it sends a release message to all other processes thereby allowing others to remove its corresponding request from their queues. 

While this algorithm is typically not viewed as a stabilizing algorithm, it can be made stabilizing with simple local checks and corrections. For example, if the queue of process $j$ contains a request from $k$, but process $k$ did not make the request then this request should be removed. The algorithm will ensure stabilization from this state, but it would still involve counters that are unbounded. 

In this program, we can observe the following: (1) As shown in Section \ref{sec:lamport}, the value of $cl.j$, the timestamp of process $j$ is a free counter. (2) Timestamps contained in any message are dependent counters. And, (3) the timestamps saved in the request queue or contained in a request/release message are dependent counters. 
\iffalse
\begin{remark}%{\bf Remark. } \ 
Note that for the algorithm in Section \ref{sec:lamport}, the value of $cl.m$ became irrelevant when $m$ was received. However, in this algorithm this value may be saved in the request queue. We treat this as creation of a new dependent counter. 
\end{remark}
\fi
\begin{remark}%{\bf Remark. } \ 
Observe that in the timestamping algorithm \cite{lamport}, the timestamp of a message became irrelevant as soon as the message was received. However, when the same timestamp was used in the mutual exclusion algorithm, even though the message timestamp became irrelevant, it was also saved in the request queue. In other words, it extended how long a dependent counter remains relevant. In other words, superimposing another program on an existing stabilizing program may increase the time for which a dependent counter is relevant thereby making it necessary to increase the bound associated with those counters. 
\end{remark}
\clearpage
\subsection{Summary of Notations}
\label{sec:SummaryofNotations}

\begin{center}
\textbf {Generic Variables}
\\[1mm]
\begin{tabular}{ |c|l| }
\hline
$p$ & program \\
\hline
$V_p$ & set of variables of program p \\
\hline
$SV_p$ & dynamic-sized equivalent of $V_p$, i.e., a\\
& dynamically changing collection of only\\ 
& simple variables, obtained by unraveling\\ 
& complex variables of $V_p$ into their constituent\\
& simple variables \\
\hline
$A_p$ & set of actions of program p\\
\hline
$s$ & state of program $p$\\
\hline
$s_l$ & $l^{th}$ state in a computation of program $p$\\
\hline
%$s_{l+1}$ & $(l+1)^{th}$ state in a computation of program $p$\\
%\hline
$guard$ & condition involving variables in $V_p$\\
\hline
$statement$ & task involving update of a subset of variables\\
& in $V_p$\\
\hline
$\rho$, $\rho'$ & computation prefixes\\
\hline
$x$ & variable in $V_p$\\
\hline
$x(s)$ & value of variable $x$ in state $s$ \\
\hline
$fc$ & free counter\\
\hline
$fc(s_l)$ & value of free counter $fc$ in state $s_l$ \\
\hline
$w$, $a$, $d$ & positive integers unless specified otherwise\\
\hline
$k_b, k_f$ & used to characterize the \textit{life} of a dependent\\ 
& counter in terms of program steps\\
%*****(combine) number of program steps (counted backward from the current step)\\
%\hline
%$k_f$ & number of program steps (counted forward from the current step)\\
\hline
$dc$ & dependent counter\\
\hline
$S$ & set of states\\
\hline
$RS$ & region size\\
\hline
$t$ & abstract global time\\
\hline
$t_j$ & physical time at process j\\
\hline
$\newfloor{\frac{t}{\regionsize}}$ & abstract global region \\
\hline
$\newfloor{\frac{t_j}{\regionsize}}$ & region of process $j$ \\
\hline
$\delta$ & duration/length of time\\
\hline
$r$ & region\\
\hline
$r_b,r_f$ & used to characterize the \textit{life} of a dependent\\ 
& counter in terms of regions\\
\hline
$max_r$ & maximum of $(r_b+r_f)$ of any dependent\\
& counter\\
\hline
$max_{inc}$ & maximum increase in any free counter within\\
& a global region\\
\hline
$p'$ & program obtained by applying our\\ 
& transformation algorithm to program $p$\\
\hline
\end{tabular}
\vfill
\textbf {Variables in Lamport's Logical Clocks example}
\\
\begin{tabular}{ |c|l| }
\hline
$j,k$ & processes\\
\hline
$cl.j$ & logical clock value of process $j$\\
\hline
$m$ & message\\
\hline
$cl.m$ & message timestamp or logical clock\\
& value associated with $m$\\
\hline
$channel_{j,k}$ & complex variable that contains timestamps of\\
& messages in transit between process $j$ and \\
&process $k$\\
\hline
$v$ & number of program steps within which a\\
&message is guaranteed to be delivered at\\
& the receiver process\\
\hline
\end{tabular}
\end{center}
\vfill
\begin{center}
\textbf{Variables in Katz and Perry example}
\begin{tabular}{ |c|l| }
\hline
$x,y$ & round number\\
\hline
$nr$ & next round\\
\hline
$cr$ & current round\\
\hline
$lr$ & round number when the last real reset was \\
& performed\\
\hline
$b$ & boolean variable that identifies if the reset was \\
& real or fake\\
\hline
\end{tabular}
\end{center}
\vspace{2mm}
\begin{center}
\textbf{Variables in Paxos based Consensus example}
\begin{tabular}{ |c|l| }
\hline
$c.seq$ & sequence number of the request made\\ 
& by proposer $c$\\
\hline
$a.seq$ & highest sequence number seen by the \\
& acceptor $a$\\
\hline
$PendingSeq$ & sequence number of pending request\\
\hline
$NextSeq$ & sequence number that would be used \\
& for a future request\\
\hline
\end{tabular}
\end{center}
\vspace{2mm}
\begin{center}
\textbf{Variables in Vector Clocks example}
\begin{tabular}{ |c|l| }
\hline
$vc.j$ & vector clock maintained at process $j$\\
\hline
$vc.j.k$ & highest clock or counter value of process $k$ \\ 
& that process $j$ is aware of\\
\hline
%\multirow{3}{4em}{Multiple row} & cell2 & cell3 \\ 
%& cell5 & cell6 \\ 
%& cell8 & cell9 \\ 
%\hline
\end{tabular}
\end{center}

\end{document}